\documentclass[a4paper]{amsart}
\usepackage{amscd}
\usepackage{amsmath}
\usepackage{amssymb}
\usepackage{mathrsfs}
\usepackage{amsthm}
\usepackage{bbm} 
\usepackage{stmaryrd}
\usepackage{tikz-cd}

\usepackage[T1]{fontenc}

\makeatletter
\@namedef{subjclassname@2020}{%
  \textup{2020} Mathematics Subject Classification}
\makeatother

\numberwithin{equation}{section} \swapnumbers

\newtheorem{satz}{Satz}[section]

\newtheorem{theorem}[satz]{Theorem}
\newtheorem{proposition}[satz]{Proposition}
\newtheorem{corollary}[satz]{Corollary}
\newtheorem{lemma}[satz]{Lemma}

\newtheorem{definition}[satz]{Definition}

\newtheorem{remark}[satz]{Remark}

\newtheorem{example}[satz]{Example}
\newtheorem{examples}[satz]{Examples}

\newcommand{\bbr}{\mathbb{R}}
\newcommand{\bbe}{\mathbb{E}}
\newcommand{\bbn}{\mathbb{N}}
\newcommand{\bbp}{\mathbb{P}}
\newcommand{\bbq}{\mathbb{Q}}
\newcommand{\bbg}{\mathbb{G}}
\newcommand{\bbf}{\mathbb{F}}

\newcommand{\bbi}{\mathbb{I}}
\newcommand{\bbk}{\mathbb{K}}

\newcommand{\bbs}{\mathbb{S}}

\newcommand{\bbx}{\mathbb{X}}

\newcommand{\cala}{\mathscr{A}}
\newcommand{\calb}{\mathscr{B}}
\newcommand{\calc}{\mathscr{C}}

\newcommand{\cale}{\mathscr{E}}
\newcommand{\calf}{\mathscr{F}}
\newcommand{\calg}{\mathscr{G}}
\newcommand{\calh}{\mathscr{H}}
\newcommand{\cali}{\mathscr{I}}
\newcommand{\calk}{\mathscr{K}}
\newcommand{\call}{\mathscr{L}}
\newcommand{\calm}{\mathscr{M}}

\newcommand{\calp}{\mathscr{P}}
\newcommand{\cals}{\mathscr{S}}
\newcommand{\calv}{\mathscr{V}}

\newcommand{\loc}{{\rm loc}}

\newcommand{\adm}{{\rm adm}}
\newcommand{\sfi}{{\rm sf}}

\newcommand{\Var}{{\rm Var}}
\newcommand{\var}{{\rm var}}

\newcommand{\lin}{{\rm lin}}

\newcommand{\la}{\langle}
\newcommand{\ra}{\rangle}

\newcommand{\bbI}{\mathbbm{1}}

\newcommand{\bdot}{\bullet}

\newcommand{\IL}{[\![}
\newcommand{\IR}{]\!]}

\begin{document}

\title[No arbitrage and multiplicative special semimartingales]{No arbitrage and multiplicative special semimartingales}
\author{Eckhard Platen \and Stefan Tappe}
\address{University of Technology Sydney, School of Mathematical and Physical Sciences, Finance Discipline Group, PO Box 123, Broadway, NSW 2007, Australia}
\email{eckhard.platen@uts.edu.au}
\address{Albert Ludwig University of Freiburg, Department of Mathematical Stochastics, Ernst-Zermelo-Stra\ss{}e 1, D-79104 Freiburg, Germany}
\email{stefan.tappe@math.uni-freiburg.de}
\date{10 September, 2022}
\thanks{We are grateful to Martin Schweizer and Josef Teichmann for valuable discussions, and to the Editor and the two referees for helpful comments and suggestions. Stefan Tappe gratefully acknowledges financial support from the Deutsche Forschungsgemeinschaft (DFG, German Research Foundation) -- project number 444121509.}
\begin{abstract}
Consider a financial market with nonnegative semimartingales which does not need to have a num\'{e}raire. We are interested in the absence of arbitrage in the sense that no self-financing portfolio gives rise to arbitrage opportunities, where we are allowed to add a savings account to the market. We will prove that in this sense the market is free of arbitrage if and only if there exists an equivalent local martingale deflator which is a multiplicative special semimartingale. In this case, the additional savings account relates to the finite variation part of the multiplicative decomposition of the deflator.
\end{abstract}

\maketitle\thispagestyle{empty}

\section{Introduction}

There exists now a rich literature on no-arbitrage concepts and their relationships. This literature involves a wide range of statements and proofs, and their links are not always obvious. Some no-arbitrage concepts remain difficult to interpret for real markets because of their purely mathematical orientation. Any form of potential arbitrage can only be exploited by market participants with the self-financing portfolio of their entire wealth. The legal concept of limited liability ensures that these portfolios remain nonnegative. Self-financing and nonnegativity of potential arbitrage portfolios are not always assumed in the no-arbitrage literature, which created a wide range of theoretical no-arbitrage concepts. However, to be practically relevant, these two properties need to be considered when clarifying whether a market model makes sense from a realistic no-arbitrage point of view.

This paper focuses on no-arbitrage concepts with the above mentioned practically relevant self-financing and nonnegativity properties, and also covers many other of the more theoretical no-arbitrage notions. It provides crucial new results concerning the links between no-arbitrage concepts and notions of deflators. Its approach leads to compact proofs and brings clarity into the links between no-arbitrage concepts. The theoretical key to these results is the handling of no-arbitrage concepts in topological vector lattices, which were developed in \cite{Platen-Tappe-tvs}. The current paper applies and extends systematically this methodology, which generalizes and systematizes the existing no-arbitrage theory.

Before giving in this introduction a brief description of the main results of the paper, let us first list some of the important papers in the no-arbitrage literature which relate to our results. These include the papers \cite{Harrison-Kreps, Harrison, Dalang, Stricker, Schach-Hilbert, Schach-94, Jacod-Shiryaev-Finance, Kabanov-Stricker} and the textbook \cite{FS}, which treat the fundamental theorem of asset pricing (FTAP) in discrete time. The papers \cite{DS-94, DS-98} and the textbook \cite{DS-book} establish the FTAP in continuous time and its connection between NFLVR (No Free Lunch with Vanishing Risk)\footnote{Concerning the no-arbitrage concepts used in this paper, we refer the reader to Definition \ref{def-NA-concepts} for the formal definitions, and to Remark \ref{rem-NA-concepts} for interpretations of these concepts.} and the existence of a martingale measure. Under the NFLVR condition the density process of a local martingale (or, more generally, a $\sigma$-martingale) measure is a strictly positive local martingale deflator. The latter is a process that acts multiplicatively and transforms nonnegative wealth processes into local martingales. The papers \cite{DS-95, Kabanov, Shiryaev-Cherny, Protter-Shimbo, Kardaras-10b, Mancin-Rung, Fontana-2, Fontana, Cuc-Teichmann, Imkeller} and the textbooks \cite{Jarrow, KK-book} treat further developments and related topics concerning the FTAP. In certain situations, results about criteria for the absence of arbitrage have been derived, for example, in \cite{Christopeit-Musiela, Melnikov-Shiryaev, Criens-1, Criens-2}.

Empirical evidence, e.g. in \cite{BGP, BIP, Platen-Rendek, SZP}, and theoretical considerations, e.g. in \cite{Loewenstein, Platen-2002, Karatzas-Kardaras}, point out that the NFLVR condition may be too restrictive for realistic long-term modeling, and that there are less expensive ways of producing long-term payouts than widely practiced under the NFLVR condition, which is somehow equivalent to risk-neutral pricing and hedging.  Several market viability properties that are weaker than NFLVR turn out to be equivalent to the existence of supermartingale (or local martingale) num\'{e}raires, where the papers \cite{Choulli-Stricker, Kardaras-12, Takaoka-Schweizer, Song} present versions of the FTAP, which connect the notions NA$_1$ (No Arbitrage of the 1st Kind), NAA$_1$ (No Asymptotic Arbitrage of the 1st Kind) and NUPBR (No Unbounded Profit with Bounded Risk) with the existence of a martingale deflator. Finally, the articles \cite{Karatzas-Kardaras, Christensen-Larsen, Hulley-Schweizer, Fontana-Rung, KKS, Herdegen, Herdegen-Schweizer-2} study related topics, including the NINA (Num\'{e}raire-Independent No-Arbitrage) condition.

So far, there are only a few references (such as \cite{Tehranchi, Herdegen, Balint-Schweizer, Balint-Schweizer-1, Balint-Schweizer-2, Harms}) which deal with financial market models without a num\'{e}raire. As a consequence, versions of the FTAP are typically formulated under the assumption that the considered market $\bbs = \{ S^1,\ldots,S^d,1 \}$ is already discounted by some num\'{e}raire. However, in \cite{Balint-Schweizer-2} it has been pointed out that, in economic terms, the existence of a discounted num\'{e}raire $S^{d+1} = 1$ is a nontrivial restriction. In the current paper, we consider a financial market $\bbs = \{ S^1,\ldots,S^d \}$ with nonnegative semimartingales. Thus, we do not need the widely used assumption that the considered market is already discounted by some num\'{e}raire, avoiding the aforementioned modeling restriction.

In our setting, a savings account does not have to exist in the market, and we demonstrate how to permit price processes and portfolios that extend the absence of arbitrage from the set of self-financing portfolios to other price and value processes. This clarifies, e.g., what can happen when adding a savings account or other price process to a market. When assuming that a savings account $B$, i.e., a predictable, strictly positive, finite variation process with $B_0 = 1$, may be added to the market, our main results can be verbally summarized as follows:
\begin{itemize}
\item The market satisfies NUPBR for its nonnegative, self-financing portfolios if and only if there exists an equivalent local martingale deflator (ELMD) $Z$ which is a multiplicative special semimartingale $Z = D B^{-1}$; see Theorem \ref{thm-FTAP}. This is a criterion which can often be checked for concrete markets; see Theorems \ref{thm-ELMD-h}, \ref{thm-ELMD-special}, \ref{thm-jd} and Corollary \ref{cor-diffusion-models}. If such a deflator exists, then its multiplicative decomposition provides us with the additional savings account $B$ which we can use in the market. This no-arbitrage characterization is also practically relevant because potential arbitrage portfolios are assumed to be self-financing and nonnegative.

\item The market satisfies NFLVR for its admissible\footnote{This means that the portfolio is allowed to be negative in between, but it must be bounded from below.}, self-financing portfolios if and only if there exists an ELMD $Z$ which is a multiplicative special semimartingale $Z = D B^{-1}$ such that the local martingale part $D$ is a true martingale; see Theorem \ref{thm-FTAP-classical}. The martingale $D$ appearing in the multiplicative decomposition $Z = D B^{-1}$ is just the density process of a measure change $\bbq \approx \bbp$, which provides a connection to the classical FTAP by Delbaen and Schachermayer (see \cite{DS-94} and \cite{DS-98}). It also provides a characterization of the currently widely practiced risk-neutral pricing approach.

\item As a consequence, when considering the discrete time setting, the market satisfies the NA (No Arbitrage) condition for all of its self-financing portfolios if and only if there exists an equivalent martingale deflator (EMD) $Z$ which is a multiplicative special semimartingale $Z = D B^{-1}$ such that the local martingale part $D$ is a true martingale; see Theorem \ref{thm-diskret-2}. Thus, the general results of this paper also allow to deduce no-arbitrage conditions for discrete time markets.
\end{itemize}
As we will see later on, an ELMD $Z$ which is a multiplicative special semimartingale does not to need exist, and if it exists, it does not need to be unique; see Examples \ref{ex-particular}. When we try to find such an ELMD $Z$, there are two possible approaches:
\begin{itemize}
\item We fix a local martingale $D > 0$ and look for a savings account $B$ such that $Z = D B^{-1}$ is an ELMD for the market $\bbs$. If such a savings account exists, then it is unique; see Proposition \ref{prop-savings-account-unique}.

\item We fix a savings account $B$ and look for a local martingale $D > 0$ such that $Z = D B^{-1}$ is an ELMD for the market $\bbs$. If such a local martingale $D$ exists, it does not need to be unique.
\end{itemize}
In general, we do not have a method in order to construct the savings account $B$ or the local martingale $D$. However, in the particular situation of diffusion models, this can be done by solving linear equations; see Corollary \ref{cor-diffusion-models}.

As already indicated, a crucial ingredient for the proofs of these results is the paper \cite{Platen-Tappe-tvs}, where we have developed a general theory of no-arbitrage concepts in topological vector lattices. More precisely, in \cite{Platen-Tappe-tvs} we have studied the relations between these no-arbitrage concepts, and we have provided abstract versions of the FTAP, including a version on Banach function spaces. For the proofs of the results from the current paper, we will apply those results from \cite{Platen-Tappe-tvs} which concern the relations between the no-arbitrage concepts; in particular the results from Section 7 in \cite{Platen-Tappe-tvs}, where the outcomes of trading strategies in a market model with semimartingales are considered, including the outcomes of nonnegative, self-financing portfolios. The abstract FTAPs from \cite{Platen-Tappe-tvs} cannot be used in order to prove the main results of the current paper.

The paper is organized as follows: In Section \ref{sec-notation} we introduce the respective notation and the no-arbitrage concepts. In Section \ref{sec-FTAP-revisited} we review FTAPs in our framework and present extensions. In Section \ref{sec-main} we describe our main results and show some of their consequences. In Section \ref{sec-tsm} we discuss consequences for semimartingale term structure models. In Section \ref{sec-ELMD} we clarify when an ELMD exists, and in Section \ref{sec-jd} we focus on jump-diffusion models with fixed times of discontinuities. In Section \ref{sec-further-examples} we provide further examples which are related to our main results, and in Section \ref{sec-dynamic}, where we treat dynamic trading strategies, we present further consequences. In Section \ref{sec-semimartingale} we discuss the assumption that the primary security accounts are semimartingales, and prove that under natural conditions this is automatically satisfied. In Section \ref{sec-filtration} we show how our results can be transferred to a mathematical setting with filtration enlargement, allowing for insider information. In Section \ref{sec-discrete} we deal with discrete time markets and present further consequences of our main results. In Appendix \ref{app-vector-integration} we provide results about vector stochastic integration, whereas in Appendix \ref{app-trans} we review a transformation result for self-financing portfolios. In Appendix \ref{app-ELMD} we present results about ELMDs and related concepts. Finally, in Appendix \ref{app-suff} we provide some sufficient conditions for the absence of arbitrage.

\section{Basic notation and no-arbitrage concepts}\label{sec-notation}

In this section we introduce basic notation and the no-arbitrage concepts which we consider in this paper, see also \cite[Sec. 7]{Platen-Tappe-tvs}. Concerning the required theory of stochastic processes, we adopt the terminology from \cite{Jacod-Shiryaev}, where further details can be found.

From now on, let $T \in (0,\infty)$ be a finite time horizon, and let $(\Omega,\calf,(\calf_t)_{t \in [0,T]},\bbp)$ be a stochastic basis satisfying the usual conditions such that $\calf_0$ is $\bbp$-trivial. We consider a market $\bbs = \{ S^i : i \in I \}$ consisting of nonnegative semimartingales $S^i \geq 0$, the primary security accounts, for some index set $I \neq \emptyset$. We denote by $\Delta(\bbs)$ the set of all strategies for $\bbs$; these are the numbers of units of components of $\bbs$ held. For a strategy $\delta \in \Delta(\bbs)$ we define the \emph{portfolio} $S^{\delta} := \delta \cdot S$, where we recall that `$\,\, \cdot \,$' denotes the usual inner product in Euclidean space. Furthermore, we denote by $\Delta_{\sfi}(\bbs)$ the set of all self-financing strategies for $\bbs$, where changes in the portfolio value result only from changes in the primary security account values. Mathematically, this is expressed by the condition $S^{\delta} = S_0^{\delta} + \delta \bdot S$, where `$\,\, \bdot \,$' denotes stochastic integration; see \cite[Sec. I.4.d]{Jacod-Shiryaev} for the stochastic integral with respect to a real-valued semimartingale, and \cite{Shiryaev-Cherny} for vector stochastic integration. For $\alpha \geq 0$ and a strategy $\delta \in \Delta(\bbs)$ we introduce the integral process $I^{\alpha,\delta} := \alpha + \delta \bdot S$. We call a process $B = (B_t)_{t \in [0,T]}$ a \emph{savings account} (or a \emph{locally risk-free asset}) if it is predictable, c\`{a}dl\`{a}g and of finite variation with $B_0 = 1$ and $B, B_- > 0$. For another nonnegative semimartingale $Y \geq 0$ we agree on the notation
\begin{align}\label{market-Y}
\bbs Y := \{ S^i Y : i \in I \}.
\end{align}
Later on, we will often consider the discounted market $\bbs B^{-1}$ for some saving account $B$. For each $\alpha \geq 0$ we introduce the following sets of potential security processes:
\begin{itemize}
\item $\mathbb{I}_{\alpha}(\bbs)$ consists of all integral processes $I^{\alpha,\delta}$ with $\delta \in \Delta(\bbs)$.

\item $\mathbb{I}_{\alpha}^{\adm}(\bbs)$ consists of all admissible integral processes from $\mathbb{I}_{\alpha}(\bbs)$. Recall that a process $X$ is called admissible if $X \geq -a$ for some constant $a \in \bbr_+$.

\item $\mathbb{I}_{\alpha}^+(\bbs)$ consists of all nonnegative integral processes from $\mathbb{I}_{\alpha}(\bbs)$.

\item $\bbp_{\sfi,\alpha}(\bbs)$ consists of all self-financing portfolios starting in $\alpha$.

\item $\bbp_{\sfi,\alpha}^{\adm}(\bbs)$ consists of all admissible self-financing portfolios starting in $\alpha$.

\item $\bbp_{\sfi,\alpha}^+(\bbs)$ consists of all nonnegative self-financing portfolios starting in $\alpha$.
\end{itemize}
In practice, under limited liability, market participants can only exploit forms of arbitrage with the nonnegative self-financing portfolio of their entire wealth, starting with their initial wealth $\alpha$. This makes the set $\bbp_{\sfi,\alpha}^+(\bbs)$ rather special from a practical perspective because it captures  natural constraints that exist for exploiting possible forms of arbitrage. Most of the other above introduced sets of security processes remain arguably more of a purely theoretical nature. Now, let $(\bbk_{\alpha})_{\alpha \geq 0}$ be any of the above families of potential security processes. Denoting by $L^0 = L^0(\Omega,\calf_T,\bbp)$ the space of all equivalence classes of real-valued random variables, where two random variables are identified if they coincide $\bbp$-almost surely, we define the family $(\calk_{\alpha})_{\alpha \geq 0}$ of subsets of $L^0$ as
\begin{align*}
\calk_{\alpha} := \{ X_T : X \in \bbk_{\alpha} \}, \quad \alpha \geq 0.
\end{align*}
We may think of outcomes of trading strategies with initial value $\alpha$. The definition above provides us with the families
\begin{align*}
&(\cali_{\alpha}(\bbs))_{\alpha \geq 0}, (\cali_{\alpha}^{\adm}(\bbs))_{\alpha \geq 0}, (\cali_{\alpha}^+(\bbs))_{\alpha \geq 0},
\\ &(\calp_{\sfi,\alpha}(\bbs))_{\alpha \geq 0}, (\calp_{\sfi,\alpha}^{\adm}(\bbs))_{\alpha \geq 0}, (\calp_{\sfi,\alpha}^+(\bbs))_{\alpha \geq 0}
\end{align*}
of outcomes of security processes. Denoting by $L_+^0$ the convex cone of all nonnegative random variables, and by $L^{\infty}$ the space of all bounded random variables, we define the convex cone $\calc \subset L^{\infty}$ as
\begin{align*}
\calc := (\calk_0 - L_+^0) \cap L^{\infty}.
\end{align*}
Moreover, we define the family $(\calb_{\alpha})_{\alpha \geq 0}$ of subsets of $L_+^0$ as
\begin{align*}
\calb_{\alpha} := (\calk_{\alpha} - L_+^0) \cap L_+^0, \quad \alpha \geq 0.
\end{align*} 
We may think of all nonnegative elements which are equal to or below the outcome of a trading strategy with initial value $\alpha$. Setting $\calb := \calb_1$, by \cite[Lemma 3.11]{Platen-Tappe-tvs} we have $\calb_{\alpha} = \alpha \calb$ for each $\alpha > 0$, and hence for the upcoming no-arbitrage concepts it suffices to consider $\calk_1$ rather than the family $(\calk_{\alpha})_{\alpha > 0}$. The Minkowski functional $p_{\calb} : L^0 \to [0,\infty]$ is given by
\begin{align*}
p_{\calb}(\xi) = \inf \{ \alpha > 0 : \xi \in \calb_{\alpha} \}, \quad \xi \in L^0.
\end{align*}
Note that $p_{\calb}(\xi)$ has the interpretation of the minimal superreplication price of $\xi$. In the following definition we denote by $L_+^{\infty}$ the convex cone of all bounded, nonnegative random variables.

\begin{definition}\label{def-NA-concepts}
We introduce the following no-arbitrage concepts:
\begin{enumerate}
\item $\calk_0$ satisfies \emph{NA (No Arbitrage)} if $\calk_0 \cap L_+^0 = \{ 0 \}$, or equivalently $\calc \cap L_+^{\infty} = \{ 0 \}$.

\item $\calk_0$ satisfies \emph{NFL (No Free Lunch)} if $\overline{\calc}^* \cap L_+^{\infty} = \{ 0 \}$, where $\overline{\calc}^*$ denotes the closure with respect to the weak-$^*$ topology $\sigma(L^{\infty},L^1)$.

\item $\calk_0$ satisfies \emph{NFLBR (No Free Lunch with Bounded Risk)} if $\widetilde{\calc}^* \cap L_+^{\infty} = \{ 0 \}$, where $\widetilde{\calc}^*$ denotes the sequential closure with respect to the weak-$^*$ topology $\sigma(L^{\infty},L^1)$.

\item $\calk_0$ satisfies \emph{NFLVR (No Free Lunch with Vanishing Risk)} if $\overline{\calc} \cap L_+^{\infty} = \{ 0 \}$, where $\overline{\calc}$ is the denotes the closure with respect to the norm topology on $L^{\infty}$.

\item $\calk_1$ satisfies \emph{NUPBR (No Unbounded Profit with Bounded Risk)} if $\calb$ is topologically bounded, or equivalently bounded in probability.

\item $\calk_1$ satisfies \emph{NAA$_1$ (No Asymptotic Arbitrage of the 1st Kind)} if $\calb$ is sequentially bounded.

\item $\calk_1$ satisfies \emph{NA$_1$ (No Arbitrage of the 1st Kind)} if $p_{\calb}(\xi) > 0$ for all $\xi \in L_+^0 \setminus \{ 0 \}$.
\end{enumerate}
\end{definition}

\begin{remark}\label{rem-NA-concepts}
Note that the above list covers several well-known no-arbitrage concepts. Recall that an \emph{arbitrage opportunity} is an element $X \in \calk_0$ such that
\begin{align*}
\bbp(X \geq 0) = 1 \quad \text{and} \quad \bbp(X > 0) > 0.
\end{align*}
Therefore, the condition $\calk_0 \cap L_+^0 = \{ 0 \}$ just means that there are no arbitrage opportunities, which explains the no-arbitrage concept NA. Often it is easy to find mathematical conditions which are sufficient for NA, but typically they fail to be necessary. In order to overcome this problem, there are two approaches in order to define slightly stronger no-arbitrage concepts:
\begin{itemize}
\item Observing that NA can equivalently be expressed as $\calc \cap L_+^{\infty} = \{ 0 \}$, we can impose stronger conditions of the type $\overline{\calc}^{\tau} \cap L_+^{\infty} = \{ 0 \}$, where the closure is taken with respect to some topology $\tau$ on $L^{\infty}$. The no-arbitrage concepts NFL, NFLBR and NFLVR introduced above are particular examples. Note that all these no-arbitrage concepts are related to $\calk_0$, which has the interpretation of outcomes of trading strategies with initial value zero.

\item Another approach is to consider no-arbitrage concepts which are related to $(\calk_{\alpha})_{\alpha > 0}$, the outcomes of trading strategies with strictly positive initial value. Then the idea is that in a reasonable market it should not be possible to make unbounded profit when investing money in that market. With our notation, this means that the set $\calb$ should be bounded in some sense, which is expressed by the no-arbitrage concepts NUPBR and NAA$_1$. The no-arbitrage concept NA$_1$ means that a strictly positive payoff $\xi \in L_+^0 \setminus \{ 0 \}$ can only be replicated with strictly positive initial wealth.
\end{itemize}
\end{remark}

The pricing under a given no-arbitrage concept is typically aligned to a, so-called, deflator. In order to introduce this concept, let $\bbx$ be a family of semimartingales. We call a semimartingale $Z = (Z_t)_{t \in [0,T]}$ such that $Z,Z_- > 0$
\begin{enumerate}
\item an \emph{equivalent martingale deflator (EMD)} for $\bbx$ if $X Z$ is a martingale for all $X \in \bbx$.

\item an \emph{equivalent local martingale deflator (ELMD)} for $\bbx$ if $X Z$ is a local martingale for all $X \in \bbx$.

\item an \emph{equivalent $\sigma$-martingale deflator (E$\mathit{\Sigma}$MD)} for $\bbx$ if $X Z$ is a $\sigma$-martingale for all $X \in \bbx$.
\end{enumerate}
A pricing rule can often also be characterized by a, so-called, pricing measure. More precisely, we call an equivalent probability measure $\bbq \approx \bbp$ on $(\Omega,\calf_T)$
\begin{enumerate}
\item an \emph{equivalent martingale measure (EMM)} for $\bbx$ if $X$ is a $\bbq$-martingale for all $X \in \bbx$.

\item an \emph{equivalent local martingale measure (ELMM)} for $\bbx$ if $X$ is a $\bbq$-local martingale for all $X \in \bbx$.

\item an \emph{equivalent $\sigma$-martingale measure (E$\mathit{\Sigma}$MM)} for $\bbx$ if $X$ is a $\bbq$-$\sigma$-martingale for all $X \in \bbx$.
\end{enumerate}
The required results about the just introduced concepts are presented in Appendix \ref{app-ELMD}. For this, an essential tool are the following two results, where Proposition \ref{prop-deflator-prop-2} will also be useful for the study of dynamic trading strategies in Section \ref{sec-dynamic} later on. For what follows, we recall that $\cals$ denotes the space of all semimartingales (see \cite[Sec. I.4.c]{Jacod-Shiryaev}), that $\calm_{\loc}$ denotes the space of all local martingales (see \cite[Sec. I.1.e]{Jacod-Shiryaev}), and that $\calm_{\sigma}$ denotes the space of all $\sigma$-martingales (see \cite[Sec. III.6.e]{Jacod-Shiryaev}).

\begin{proposition}\label{prop-deflator-prop-1}
Let $X \in \cals^d$ and $Y \in \calm_{\sigma}$ be such that $X^i Y \in \calm_{\sigma}$ for each $i = 1,\ldots,d$. Then for every $H \in L(X)$ we have $(H \bdot X) Y \in \calm_{\sigma}$.
\end{proposition}

\begin{proof}
Let $i \in \{ 1,\ldots,d \}$ be arbitrary. Using integration by parts (see \cite[Def. I.4.45]{Jacod-Shiryaev}) we have
\begin{align*}
X_0^i Y_0 + X_-^i \bdot Y + Y_- \bdot X^i + [X^i,Y] = X^i Y \in \calm_{\sigma}.
\end{align*}
Since $Y \in \calm_{\sigma}$, we have $X_-^i \bdot Y \in \calm_{\sigma}$, and hence
\begin{align*}
Y_- \bdot X^i + [X^i,Y] \in \calm_{\sigma}.
\end{align*}
By Lemma \ref{lemma-bracket} we have $H \in L([X,Y])$ and
\begin{align*}
[H \bdot X,Y] = H \bdot [X,Y].
\end{align*}
Since $Y_-$ is predictable and locally bounded, by Lemma \ref{lemma-ass} we have
\begin{align*}
Y_- \in L(H \bdot X), \quad Y_- H \in L(X), \quad H \in L(Y_- \bdot X)
\end{align*}
and
\begin{align*}
Y_- \bdot (H \bdot X) = (Y_- H) \bdot X = H \bdot (Y_- \bdot X).
\end{align*}
Therefore, using integration by parts (see \cite[Def. I.4.45]{Jacod-Shiryaev}) again, we obtain
\begin{align*}
(H \bdot X) Y &= (H \bdot X)_- \bdot Y + Y_- \bdot (H \bdot X) + [H \bdot X,Y]
\\ &= (H \bdot X)_- \bdot Y + H \bdot (Y_- \bdot X) + H \bdot [X,Y]
\\ &= (H \bdot X)_- \bdot Y + H \bdot ( Y_- \bdot X + [X,Y] ) \in \calm_{\sigma},
\end{align*}
completing the proof.
\end{proof}

\begin{proposition}\label{prop-deflator-prop-2}
Let $X \in \cals^d$ and $Z \in \cals$ be such that $X^i Z \in \calm_{\sigma}$ for each $i = 1,\ldots,d$. Then for every $H \in L(X)$ and every $Y \in \calm_{\sigma}$ with $[Y,Z] = 0$ and
\begin{align}\label{sf-app}
H \cdot X = H_0 \cdot X_0 + H \bdot X + Y
\end{align}
the process $Y$ is predictable, and we have $(H \cdot X) Z \in \calm_{\sigma}$.
\end{proposition}

\begin{proof}
By (\ref{sf-app}) we have $H \cdot X \in \cals$. Using integration by parts (see \cite[Def. I.4.45]{Jacod-Shiryaev}) we have
\begin{equation}\label{sf-prod-1}
\begin{aligned}
(H \cdot X - Y) Z &= (H_0 \cdot X_0 - Y_0) Z_0 + (H \cdot X - Y)_- \bdot Z
\\ &\quad + Z_- \bdot (H \cdot X - Y ) + [H \cdot X - Y,Z].
\end{aligned}
\end{equation}
By (\ref{sf-app}) and Lemma \ref{lemma-bracket} we have $H \in L_{\rm var}([X,Z])$ and
\begin{align}\label{sf-prod-1a}
[H \cdot X - Y,Z] = [H \bdot X,Z] = H \bdot [X,Z].
\end{align}
Furthermore, we have
\begin{align*}
(H \cdot X - Y)_- &= (H \cdot X - Y) - \Delta ( H \cdot X - Y ) = ( H \cdot X - Y ) - \Delta ( H \bdot X )
\\ &= ( H \cdot X - Y ) - H \cdot \Delta X = H \cdot X_- - Y.
\end{align*}
Therefore, we have
\begin{align*}
\Delta Y = H \cdot X_- - (H \cdot X)_-, 
\end{align*}
showing that the process $Y$ is predictable. By \cite[Thm. I.4.52]{Jacod-Shiryaev}, the assumption $[Y,Z] = 0$ implies that
\begin{align*}
\sum_{s \leq t} \Delta Y_s \Delta Z_s = 0, \quad t \in \bbr_+,
\end{align*}
and hence we have $\Delta Y \in L(Z)$ with $\Delta Y \bdot Z = 0$. Therefore, we also have $Y \in L(Z)$ with $Y \bdot Z = Y_- \bdot Z$. Furthermore, by Lemma \ref{lemma-int-same-comp} and \cite[Thm. 4.7]{Shiryaev-Cherny} we have $H X_- \in L(Z \bbI_{\bbr^d})$, $H \in L(X_- \bdot Z)$ and
\begin{equation}\label{sf-prod-2}
\begin{aligned}
(H \cdot X - Y)_- \bdot Z &= (H \cdot X_- - Y) \bdot Z = (H X_-) \bdot (Z \bbI_{\bbr^d}) - Y \bdot Z
\\ &= H \bdot (X_- \bdot Z) - Y_- \bdot Z,
\end{aligned}
\end{equation}
where $X_- \bdot Z$ denotes the $\bbr^d$-valued process with components $(X_- \bdot Z)^i = X_-^i \bdot Z$ for each $i=1,\ldots,d$.
Furthermore, by Lemma \ref{lemma-ass} we have $H \in L(Z_- \bdot X)$ and
\begin{align}\label{sf-prod-3}
Z_- \bdot (H \cdot X - Y) = Z_- \bdot (H \bdot X) = H \bdot (Z_- \bdot X),
\end{align}
where $Z_- \bdot X$ denotes the $\bbr^d$-valued process with components $(Z_- \bdot X)^i = Z_- \bdot X^i$ for each $i=1,\ldots,d$.  Consequently, using (\ref{sf-prod-1})--(\ref{sf-prod-3}), integration by parts and the assumption $[Y,Z] = 0$, we deduce that
\begin{align*}
(H \cdot X) Z &= (H \cdot X - Y) Z + YZ
\\ &= (H_0 \cdot X_0 - Y_0) Z_0 + H \bdot \big( X_- \bdot Z + Z_- \bdot X + [X,Z] \big) - Y_- \bdot Z + YZ
\\ &= (H_0 \cdot X_0) Z_0 + H \bdot (XZ) + Z_- \bdot Y \in \calm_{\sigma},
\end{align*}
where $XZ \in \calm_{\sigma}^d$ denotes the $\bbr^d$-valued $\sigma$-martingale with components $X^i Z$ for each $i = 1,\ldots,d$.
\end{proof}

\section{The fundamental theorems of asset pricing revisited}\label{sec-FTAP-revisited}

In this section we review the fundamental theorems of asset pricing in our present framework for the given no-arbitrage concepts, and present some extensions. The mathematical framework is that of the previous Section \ref{sec-notation}, however, now we consider a finite market $\bbs = \{ S^1,\ldots,S^d \}$ with nonnegative semimartingales for some $d \in \bbn$. Let us first formulate a theorem that links weaker no-arbitrage concepts related to NUPBR that allow a richer modelling world than those we link later on to NFLVR.

\begin{theorem}\label{thm-Takaoka}
The following statements are equivalent:
\begin{enumerate}
\item[(i)] $\cali_1^+(\bbs)$ satisfies NUPBR. 

\item[(ii)] $\cali_1^+(\bbs)$ satisfies NAA$_1$.

\item[(iii)] $\cali_1^+(\bbs)$ satisfies NA$_1$.

\item[(iv)] There exists an ELMD $Z$ for $\bbs$ such that $Z \in \calm_{\loc}$.
\end{enumerate}
If the previous conditions are fulfilled, then $\cali_0^+(\bbs)$ satisfies NFL, NFLBR, NFLVR and NA.
\end{theorem}

Recall that the securities in the set $\cali_1^+(\bbs)$ represent the nonnegative sums of the initial security value one and integrals of strategies with respect to primary security accounts. Important is that NUPBR is equivalent to the existence of an ELMD which is a local martingale. 

\begin{proof}[Proof of Theorem \ref{thm-Takaoka}]
(i) $\Leftrightarrow$ (ii) $\Leftrightarrow$ (iii): See, for example \cite[Thm. 7.25]{Platen-Tappe-tvs}.

\noindent(i) $\Leftrightarrow$ (iv): Noting that a process $Z$ with $Z,Z_- > 0$ is an ELMD $Z$ for $\bbs$ with $Z \in \calm_{\loc}$ if and only if it is a strict $\sigma$-martingale density in the sense of \cite{Takaoka-Schweizer}, this equivalence follows from \cite[Thm. 2.6]{Takaoka-Schweizer}.

\noindent The additional statement is a consequence of Proposition \ref{prop-ELMD-suff} and \cite[Prop. 5.7]{Platen-Tappe-tvs}.
\end{proof}

Let us now link no-arbitrage concepts related to the stronger widely assumed NFLVR condition.

\begin{theorem}\label{thm-FTAP-DS}
The following statements are equivalent:
\begin{enumerate}
\item[(i)] $\cali_0^{\adm}(\bbs)$ satisfies NFL.

\item[(ii)] $\cali_0^{\adm}(\bbs)$ satisfies NFLBR.

\item[(iii)] $\cali_0^{\adm}(\bbs)$ satisfies NFLVR.

\item[(iv)] There exists an ELMD $Z$ for $\bbs$ such that $Z$ is a martingale.

\item[(v)] There exists an ELMM $\bbq \approx \bbp$ on $(\Omega,\calf_T)$ for $\bbs$.
\end{enumerate}
If the previous conditions are fulfilled, then $\cali_0^{\adm}(\bbs)$ satisfies NA, and $\cali_1^{\adm}(\bbs)$ satisfies NUPBR, NAA$_1$ and NA$_1$.
\end{theorem}

Recall that the set $\cali_0^{\adm}(\bbs)$ relates to admissible securities that are the  integral of strategies with respect to primary security accounts and can get negative down to a certain level. Important is that NFLVR is equivalent to the existence of a deflator that is a true martingale, which underpins risk-neutral pricing.

\begin{proof}[Proof of Theorem \ref{thm-FTAP-DS}]
(i) $\Rightarrow$ (ii) $\Rightarrow$ (iii): See, for example \cite[Prop. 5.7]{Platen-Tappe-tvs}.

\noindent(iv) $\Leftrightarrow$ (v): This equivalence is obvious.

\noindent(iii) $\Leftrightarrow$ (v): Noting that the semimartingales $S^i$, $i=1,\ldots,d$ are nonnegative, by Lemma \ref{lemma-M-sigma-loc} an equivalent probability measure $\bbq \approx \bbp$ is an ELMM for $\bbs$ if and only if it is an E$\Sigma$MM for $\bbs$. Therefore, the stated equivalence is a consequence of \cite[Thm. 1.1]{DS-98}.

\noindent(v) $\Rightarrow$ (i): This implication follows from Proposition \ref{prop-ELMD-suff}.

\noindent The additional statements follow from \cite[Prop. 5.7, Prop. 7.27 and Thm. 7.25]{Platen-Tappe-tvs}.
\end{proof}

\section{The main results and their consequences}\label{sec-main}

In this section we present our main results and show some of their consequences. We emphasize that the proofs of the main results presented in this section are straightforward, and essentially rely on the FTAPs revisited in Section \ref{sec-FTAP-revisited}. As in Section \ref{sec-FTAP-revisited}, we consider a finite market $\bbs = \{ S^1,\ldots,S^d \}$ with nonnegative semimartingales for some $d \in \bbn$. Recall that a semimartingale $Z$ with $Z,Z_- > 0$ is called a \emph{multiplicative special semimartingale} if it admits a multiplicative decomposition
\begin{align}\label{mult-decomp}
Z = DC
\end{align}
with a local martingale $D > 0$ and a predictable c\`{a}dl\`{a}g process $C > 0$ of finite variation. Then by \cite[Lemma III.3.6]{Jacod-Shiryaev} we also have $D_- > 0$, and hence $C_- > 0$.

\begin{theorem}\cite[Thm. II.8.21]{Jacod-Shiryaev}\label{thm-mult-decomp}
For a semimartingale $Z$ with $Z,Z_- > 0$ the following statements are equivalent:
\begin{enumerate}
\item[(i)] $Z$ is a multiplicative special semimartingale.

\item[(ii)] $Z$ is a special semimartingale.
\end{enumerate}
If the previous conditions are fulfilled, then the processes $D$ and $C$ appearing in the multiplicative decomposition (\ref{mult-decomp}) are unique up to an evanescent set.
\end{theorem}

In the situation of Theorem \ref{thm-mult-decomp} we call $D$ the \emph{local martingale part} and $C$ the \emph{finite variation part} of the multiplicative decomposition (\ref{mult-decomp}).

\begin{lemma}\label{lemma-extended-B}
Let $Z = D B^{-1}$ be a multiplicative special semimartingale with a local martingale part $D \in \calm_{\loc}$ and a savings account $B$. Then the following statements are equivalent:
\begin{enumerate}
\item[(i)] $Z$ is an ELMD for $\bbs$.

\item[(ii)] $Z$ is an ELMD for $\bbs \cup \{ B \}$.
\end{enumerate}
\end{lemma}

\begin{proof}
Since $BZ = D \in \calm_{\loc}$, the proof is immediate.
\end{proof}

Now, we are ready to state our result about the no-arbitrage concept NUPBR under the for practice important two natural constraints of nonnegativity and self-financing, which give access to a wide range of financial market models.

\begin{theorem}\label{thm-FTAP}
The following statements are equivalent:
\begin{enumerate}
\item[(i)] There exists a savings account $B$ such that $\calp_{\sfi,1}^+(\bbs \cup \{ B \})$ satisfies NUPBR.

\item[(ii)] There exists a savings account $B$ such that $\calp_{\sfi,1}^+(\bbs \cup \{ B \})$ satisfies NAA$_1$.

\item[(iii)] There exists a savings account $B$ such that $\calp_{\sfi,1}^+(\bbs \cup \{ B \})$ satisfies NA$_1$.

\item[(iv)] There exist a savings account $B$ and an ELMD $D$ for $\bbs B^{-1}$ such that $D \in \calm_{\loc}$.

\item[(v)] There exists an ELMD $Z$ for $\bbs$ which is a multiplicative special semimartingale.
\end{enumerate}
If the previous conditions are fulfilled, then the savings accounts $B$ in (i)--(iv) can be chosen to be equal, and in (v) we can choose an ELMD $Z$ for $\bbs$ with multiplicative decomposition $Z = D B^{-1}$ with this savings account $B$. Furthermore $\calp_{\sfi,0}^+(\bbs \cup \{ B \})$ satisfies NA.
\end{theorem}

\begin{proof}
Let $B$ be an arbitrary savings account.

\noindent(i) $\Leftrightarrow$ (ii) $\Leftrightarrow$ (iii): See, for example \cite[Thm. 7.25]{Platen-Tappe-tvs}.

\noindent(i) $\Rightarrow$ (iv): By Proposition \ref{prop-NA-concepts-equiv} the set $\cali_1^+(\bbs B^{-1})$ satisfies NUPBR. Hence, by Theorem \ref{thm-Takaoka} there exists an ELMD $D$ for $\bbs B^{-1}$ such that $D \in \calm_{\loc}$.

\noindent(iv) $\Rightarrow$ (v): It is obvious that $Z = D B^{-1}$ is an ELMD for $\bbs$.

\noindent(v) $\Rightarrow$ (i): By Lemma \ref{lemma-extended-B} the process $Z$ is also an ELMD for $\bbs \cup \{ B \}$. Hence, by Proposition \ref{prop-ELMD-suff} the set $\calp_{\sfi,1}^+(\bbs \cup \{ B \})$ satisfies NUPBR.

\noindent The additional statement follows from \cite[Prop. 7.28]{Platen-Tappe-tvs}.
\end{proof}

Important in the previous result is that NUPBR is equivalent to NAA$_1$ and NA$_1$, as well as the existence of a deflator which is a multiplicative special semimartingale.

The previous result is also in line with the no-arbitrage concept \emph{Num\'{e}raire-Independent No-Arbitrage (NINA)} from \cite{Herdegen}, which means (in the mathematical framework of \cite{Herdegen}) that the zero strategy is strongly maximal for the strategy cone of all self-financing, undefaultable strategies; see \cite[Def. 3.21]{Herdegen} for the precise definition. Verbally, the definition says that every nonzero (nonnegative) contingent claim must have a positive (superreplication) price. Moreover, as the name NINA suggests, it does not assume the existence of a num\'{e}raire strategy; see the paragraphs after \cite[Def. 3.21]{Herdegen}, where the relations to other no-arbitrage concepts are discussed as well. For the proof of the following result we will use \cite[Thm. 4.10]{Herdegen}, where NINA is characterized for so-called num\'{e}raire markets. In (\ref{market-NINA}) below, $\cals$ denotes the space of all semimartingales.

\begin{theorem}
The following statements are equivalent:
\begin{enumerate}
\item[(i)] There exists a savings account $B$ such that the market 
\begin{align}\label{market-NINA}
\{ (S^1 Z,\ldots,S^d Z,B Z) : Z \in \cals \text{ with } Z,Z_- > 0  \}
\end{align}
satisfies NINA.

\item[(ii)] There exists an ELMD $Z$ for $\bbs$ which is a multiplicative special semimartingale.
\end{enumerate}
If the previous conditions are fulfilled, then we can choose an ELMD $Z$ for $\bbs$ with multiplicative decomposition $Z = D B^{-1}$, where $B$ is a savings account as in (i).
\end{theorem}

\begin{proof}
First of all, note that for every saving account $B$ we have
\begin{align*}
\inf_{0 \leq t \leq T} \bigg( \sum_{i=1}^d |S_t^i| + |B_t| \bigg) \geq \inf_{0 \leq t \leq T} B_t > 0 \quad \text{$\bbp$-almost surely,}
\end{align*}
showing that (\ref{market-NINA}) is a market in the sense of \cite[Def. 2.3]{Herdegen}. Furthermore, the market (\ref{market-NINA}) is a num\'{e}raire market in the sense of \cite[Def. 2.10]{Herdegen} because $e_{d+1} = (0,\ldots,0,1)$ is a num\'{e}raire strategy. This ensures that we may apply \cite[Thm. 4.10]{Herdegen} in the sequel.

\noindent(i) $\Rightarrow$ (ii): By \cite[Thm. 4.10]{Herdegen} there exists a semimartingale $Z$ with $Z,Z_- > 0$ such that $(S^1 Z,\ldots,S^d Z,B Z)$ is a $\bbp$-local-martingale representative; that is, we have $S^1 Z,\ldots,S^d Z,B Z \in \calm_{\loc}$, showing that $Z$ is an ELMD for $\bbs$. Furthermore, the process $Z$ is a multiplicative special semimartingale with multiplicative decomposition $Z = D B^{-1}$, where $D = BZ$.

\noindent(ii) $\Rightarrow$ (i): There exist a local martingale $D \in \calm_{\loc}$ with $D > 0$ and a savings account $B$ such that $Z = D B^{-1}$. Since $Z$ is an ELMD for $\bbs$, we have $S^1 Z, \ldots, S^d Z, BZ \in \calm_{\loc}$; that is, the process $(S^1 Z, \ldots, S^d Z, BZ)$ is $\bbp$-local-martingale representative in the terminology of \cite{Herdegen}. Therefore, by \cite[Thm. 4.10]{Herdegen} the market (\ref{market-NINA}) satisfies NINA.
\end{proof}

\begin{remark}
A semimartingale $Z$ with $Z,Z_- > 0$ as in (\ref{market-NINA}) is also called an \emph{exchange rate process}; see \cite[Def. 2.1]{Herdegen}.
\end{remark}

Next, we present our result concerning the no-arbitrage concept NFLVR, which is more restrictive than NUPBR.

\begin{theorem}\label{thm-FTAP-classical}
The following statements are equivalent:
\begin{enumerate}
\item[(i)] There exists a savings account $B$ such that $B$ and $B^{-1}$ are bounded, and $\calp_{\sfi,0}^{\adm}(\bbs \cup \{ B \})$ satisfies NFL.

\item[(ii)] There exists a savings account $B$ such that $B$ and $B^{-1}$ are bounded, and $\calp_{\sfi,0}^{\adm}(\bbs \cup \{ B \})$ satisfies NFLBR.

\item[(iii)] There exists a savings account $B$ such that $B$ and $B^{-1}$ are bounded, and $\calp_{\sfi,0}^{\adm}(\bbs \cup \{ B \})$ satisfies NFLVR.

\item[(iv)] There exist a savings account $B$ such that $B$ and $B^{-1}$ are bounded, and an ELMM $\bbq \approx \bbp$ for $\bbs B^{-1}$.

\item[(v)] There exists an ELMD $Z$ for $\bbs$ which is a multiplicative special semimartingale such that the local martingale part is a true martingale, and the finite variation part and its inverse are bounded.
\end{enumerate}
If the previous conditions are fulfilled, then the savings accounts $B$ in (i)--(iv) can be chosen to be equal, and in (v) we can choose an ELMD $Z$ for $\bbs$ with multiplicative decomposition $Z = D B^{-1}$ with this savings account $B$. Furthermore $\calp_{\sfi,0}^{\adm}(\bbs \cup \{ B \})$ satisfies NA, and $\calp_{\sfi,1}^{\adm}(\bbs \cup \{ B \})$ satisfies NA$_1$, NAA$_1$ and NUPBR.
\end{theorem}

\begin{proof}
Let $B$ be a savings account $B$ such that $B$ and $B^{-1}$ are bounded.

\noindent(i) $\Rightarrow$ (ii) $\Rightarrow$ (iii): See, for example \cite[Prop. 5.7]{Platen-Tappe-tvs}.

\noindent(iii) $\Rightarrow$ (v): By Proposition \ref{prop-NA-concepts-equiv} the set $\cali_0^{\adm}(\bbs B^{-1})$ satisfies NFLVR. Hence, by Theorem \ref{thm-FTAP-DS} there exists an ELMD $D$ for $\bbs B^{-1}$ such that $D$ is a martingale. Therefore, the process $Z = D B^{-1}$ is an ELMD for $\bbs$.

\noindent(v) $\Rightarrow$ (iv): Note that $D$ is an ELMD for $\bbs B^{-1}$. Let $\bbq \approx \bbp$ be the equivalent probability measure on $(\Omega,\calf_T)$ with Radon-Nikodym derivative $\frac{d \bbq}{d \bbp} = D_T / D_0$. Then $\bbq$ is an ELMM for $\bbs B^{-1}$.

\noindent(iv) $\Rightarrow$ (i): By Proposition \ref{prop-ELMD-suff} the set $\cali_0^{\adm}(\bbs B^{-1})$ satisfies NFL. Thus, by Proposition \ref{prop-NA-concepts-equiv} the set $\calp_{\sfi,0}^{\adm}(\bbs \cup \{ B \})$ satisfies NFL as well.

\noindent The remaining statements follow from \cite[Prop. 5.7, Prop. 7.29 and Thm. 7.25]{Platen-Tappe-tvs}.
\end{proof}

\begin{remark}
In Theorem  \ref{thm-FTAP-classical} the boundedness assumption concerning the savings account is needed since we use Proposition \ref{prop-NA-concepts-equiv} concerning market transformations for the proof. We refer to Remark \ref{remark-B-bounded}, where the use of the boundedness assumption is explained.
\end{remark}

Important in the previous result to note is that NFLVR is equivalent to the existence of an ELMM $\bbq \approx \bbp$, which underpins the practice of risk-neutral pricing. Furthermore, when NFLVR is requested, then NUPBR still holds. However, this is, in general, not true the other way around. Indeed, if there exists a deflator as in Theorem \ref{thm-FTAP}, which is a multiplicative special semimartingale, then NFLVR only holds true if the local martingale part is a true martingale, which then gives rise to the aforementioned measure change. This means that one is not giving anything away in modeling freedom when assuming NUPBR, but restricts the phenomena one can describe, when postulating NFLVR for a market model, which is what most of the literature has been assuming.

Before we proceed with the question how to extend a market by a savings account, let us prepare an auxiliary result.

\begin{lemma}\label{lemma-prod-M-A}
Let $M \in \calm_{\loc}$ be a local martingale with $M > 0$, and let $A$ be a predictable c\`{a}dl\`{a}g process of finite variation such that $MA \in \calm_{\loc}$. Then we have $A = A_0$ up to an evanescent set.
\end{lemma}

\begin{proof}
Using integration by parts (see \cite[Def. I.4.45]{Jacod-Shiryaev}) we have
\begin{align*}
MA = M_0 A_0 + M_- \bdot A + A_- \bdot M + [M,A].
\end{align*}
By \cite[Prop. I.4.49.c]{Jacod-Shiryaev} we have $[M,A] \in \calm_{\loc}$, and hence $M_- \bdot A \in \calm_{\loc} \cap \calv$. Since this process is also predictable, by \cite[Cor. I.3.16]{Jacod-Shiryaev} we deduce that $M_- \bdot A = 0$. Since $M > 0$, by \cite[Lemma III.3.6]{Jacod-Shiryaev} we also have $M_- > 0$, and it follows that $A = A_0$ up to an evanescent set.
\end{proof}

As we have seen, NUPBR represents, in the described sense, the least restrictive no-arbitrage concept, and we may assume it for our next result. In the previous results, the savings account $B$ could already be contained in the market $\bbs$. As the next result shows, an arbitrage free market can have at most one savings account.

\begin{proposition}
Suppose that $\calp_{\sfi,1}^+(\bbs)$ satisfies NUPBR (or, equivalently, NAA$_1$ or NA$_1$), and let $B,\hat{B} \in \bbs$ be two savings accounts. Then we have $B = \hat{B}$ up to an evanescent set.
\end{proposition}

\begin{proof}
By Theorem \ref{thm-FTAP} there exists a local martingale $D \in \calm_{\loc}$ with $D > 0$ such that $Z = D B^{-1}$ is an ELMD for $\bbs$. In particular, setting $A := B^{-1} \hat{B}$ we have $DA = Z \hat{B} \in \calm_{\loc}$. Applying Lemma \ref{lemma-prod-M-A} gives us $A=1$ up to an evanescent set, and hence $B = \hat{B}$ up to an evanescent set.
\end{proof}

In the situation of the previous results, the savings account $B$, and hence the ELMD $Z = D B^{-1}$, do not need to be unique. However, as the following result shows, for a given local martingale $D > 0$ there is at most one suitable savings account fitting into the multiplicative decomposition $Z = D B^{-1}$. For this result, we consider a general market $\bbs = \{ S^i : i \in I \}$ with an arbitrary index set $I \neq \emptyset$.

\begin{proposition}\label{prop-savings-account-unique}
Suppose that $S^i > 0$ for some $i \in I$. Let $D > 0$ be a local martingale, and let $B,\hat{B}$ be two savings accounts such that the multiplicative special semimartingales $Z = D B^{-1}$ and $\hat{Z} = D \hat{B}^{-1}$ are ELMDs for the market $\bbs$. Then we have $B = \hat{B}$ up to an evanescent set.
\end{proposition}

\begin{proof}
We have $S^i D B^{-1} \in \calm_{\loc}$ and $S^i D \hat{B}^{-1} \in \calm_{\loc}$. Note that $A := B \hat{B}^{-1}$ is another savings account, and that
\begin{align*}
(S^i D B^{-1}) A \in \calm_{\loc}.
\end{align*}
Applying Lemma \ref{lemma-prod-M-A} gives us that $A = 1$ up to an evanescent set, and hence we have $B = \hat{B}$ up to an evanescent set.
\end{proof}

\begin{remark}\label{rem-more-general}
In this section we have considered a finite market $\bbs = \{ S^1,\ldots,S^d \}$. However, note that for an arbitrary market $\bbs = \{ S^i : i \in I \}$ with nonnegative semimartingales and an arbitrary index set $I \neq \emptyset$ the existence of an appropriate ELMD, which is a multiplicative special semimartingale, is sufficient for the absence of arbitrage. More precisely, in such a more general market the following implications still hold true:
\begin{itemize}
\item (iv) $\Rightarrow$ (i), (v) $\Rightarrow$ (i) in Theorem \ref{thm-FTAP}.

\item (iv) $\Rightarrow$ (i), (v) $\Rightarrow$ (i) in Theorem \ref{thm-FTAP-classical}.
\end{itemize}
\end{remark}

\section{Semimartingale term structure models}\label{sec-tsm}

In this section we discuss the connections with the results of \cite{Doeberlein-Schweizer} concerning semimartingale term structure models. We consider a market $\bbs = \{ P^S : S \in [0,T] \}$ consisting of zero coupon bonds. Assume that $P^S,P_-^S > 0$ and $P_S^S = 1$ for each $S \in [0,T]$. We call the market $\bbs$ a \emph{semimartingale term structure model}, or simply a \emph{term structure model}. The following concepts of a generated term structure model and of an implied savings account are provided in \cite{Doeberlein-Schweizer} in a risk-neutral framework. We recall and extend these notions to the present situation with the real-world probability measure as follows.

\begin{definition}
We introduce the following concepts:
\begin{enumerate}
\item Let $\bbq \approx \bbp$ be an equivalent probability measure on $(\Omega,\calf_T)$, and let $Y$ be a semimartingale such that $Y,Y_- > 0$. We say that the term structure model $\bbs$ is \emph{generated by $(\bbq,Y)$} if $\bbq$-almost surely
\begin{align*}
P_s^S = \bbe_{\bbq} \bigg[ \frac{Y_S}{Y_s} \bigg| \calf_s \bigg] \quad \text{for all $0 \leq s \leq S \leq T$.}
\end{align*}
\item Let $Z$ be a semimartingale such that $Z,Z_- > 0$. We say that the term structure model $\bbs$ is generated by $Z$ if $\bbp$-almost surely
\begin{align*}
P_s^S = \bbe \bigg[ \frac{Z_S}{Z_s} \bigg| \calf_s \bigg] \quad \text{for all $0 \leq s \leq S \leq T$.}
\end{align*}
\end{enumerate}
\end{definition}

\begin{definition}
Let $B$ be a savings account.
\begin{enumerate}
\item Let $\bbr \approx \bbp$ be an equivalent probability measure on $(\Omega,\calf_T)$. We say that $B$ is a \emph{savings account implied by the term structure model $\bbs$ relative to the measure $\bbr$} if $\bbr$-almost surely
\begin{align*}
P_s^S = \bbe_{\bbr} \bigg[ \frac{B_S^{-1}}{B_s^{-1}} \bigg| \calf_s \bigg] \quad \text{for all $0 \leq s \leq S \leq T$.}
\end{align*}

\item Let $D > 0$ be a local martingale. We say that $B$ is a \emph{savings account implied by the term structure model $\bbs$ relative to the local martingale $D$} if $\bbp$-almost surely
\begin{align*}
P_s^S = \bbe \bigg[ \frac{D_S B_S^{-1}}{D_s B_s^{-1}} \bigg| \calf_s \bigg] \quad \text{for all $0 \leq s \leq S \leq T$.}
\end{align*}
\end{enumerate}
\end{definition}

\begin{remark}\label{rem-bond-NA}
Let $B$ be a savings account.
\begin{enumerate}
\item Suppose that $B$ is implied by the term structure model $\bbs$ relative to some local martingale $D > 0$. Then the multiplicative special semimartingale $Z := D B^{-1}$ is an ELMD for the term structure model $\bbs$, and hence, by Remark \ref{rem-more-general} the set $\calp_{\sfi,1}^+(\bbs \cup \{ B \})$ satisfies NUPBR.

\item Suppose that $B$ is implied by the term structure model $\bbs$ relative to some equivalent probability measure $\bbr \approx \bbp$ on $(\Omega,\calf_T)$. Then the measure $\bbr$ is an ELMM for the discounted term structure model $\bbs B^{-1}$, and hence, by Remark \ref{rem-more-general} the set $\calp_{\sfi,0}^{\adm}(\bbs \cup \{ B \})$ satisfies NFL.
\end{enumerate}
\end{remark}

\begin{remark}\label{rem-implied-savings-1}
Let $Z = D B^{-1}$ be a multiplicative special semimartingale with a local martingale $D > 0$ and a savings account $B$ such that the term structure model $\bbs$ is generated by $Z$. Then the following statements are true:
\begin{enumerate}
\item $B$ is a savings account implied by $\bbs$ relative to $D$.

\item If $D$ is a true martingale, then using \cite[III.3.9]{Jacod-Shiryaev} shows that $B$ is a savings account implied by $\bbs$ relative to $\bbr$, where $\bbr \approx \bbp$ denotes the equivalent probability measure on $(\Omega,\calf_T)$ with density process $D/D_0$.
\end{enumerate}
\end{remark}

\begin{proposition}\label{prop-implied-savings-2}
Let $B$ be a savings account, let $\bbq \approx \bbp$ be an equivalent probability measure on $(\Omega,\calf_T)$, and let $M > 0$ be a $\bbq$-local martingale. If $\bbs$ is generated by $(\bbq,Y)$, where $Y := M B^{-1}$, then the following statements are true:
\begin{enumerate}
\item $B$ is a savings account implied by $\bbs$ relative to the local martingale $MD$, where $D$ denotes the density process of the measure change $\bbq \approx \bbp$.

\item If $M$ is a true $\bbq$-martingale, then $B$ is a savings account implied by $\bbs$ relative to $\bbr$, where $\bbr \approx \bbq$ denotes the equivalent probability measure on $(\Omega,\calf_T)$ with density process $M/M_0$.
\end{enumerate}
\end{proposition}

\begin{proof}
By \cite[III.3.9]{Jacod-Shiryaev} the term structure model $\bbs$ is generated by $Z := M D B^{-1}$, and by \cite[Prop. III.3.8]{Jacod-Shiryaev} the process $MD > 0$ is a local martingale. Therefore, by Remark \ref{rem-implied-savings-1} the process $B$ is a savings account implied by $\bbs$ relative to $MD$, proving the first statement.

For the proof of the second statement we assume that $M$ is a true $\bbq$-martingale. By \cite[Prop. III.3.8]{Jacod-Shiryaev} the process $MD$ is a true martingale. Therefore, by \cite[III.3.9]{Jacod-Shiryaev} the process $B$ is a savings account implied by $\bbs$ relative to $\bbr$, where $\bbr \approx \bbp$ denotes the equivalent probability measure on $(\Omega,\calf_T)$ with density process $MD/M_0$. To complete the proof, note that $\bbr \approx \bbq$ with density process $M/M_0$.
\end{proof}

\begin{remark}
The second statement of Proposition \ref{prop-implied-savings-2} provides the first statement of \cite[Thm. 5]{Doeberlein-Schweizer}. In the terminology of \cite{Doeberlein-Schweizer}, the pair $(\bbq,Y)$ is called \emph{good}.
\end{remark}

We close this section with the following result about the uniqueness of an implied savings account.

\begin{proposition}
Let $B$ and $\hat{B}$ be two savings accounts. Then the following statements are true:
\begin{enumerate}
\item Let $D > 0$ be a local martingale such that $B$ and $\hat{B}$ are savings accounts implied by $\bbs$ relative to $D$. Then we have $B = \hat{B}$ up to an evanescent set.

\item Let $\bbr \approx \bbp$ be an equivalent probability measure on $(\Omega,\calf_T)$ such that $B$ and $\hat{B}$ are savings accounts implied by $\bbs$ relative to $\bbr$. Then we have $B = \hat{B}$ up to an evanescent set.
\end{enumerate}
\end{proposition}

\begin{proof}
This is a consequence of Remark \ref{rem-bond-NA} and Proposition \ref{prop-savings-account-unique}.
\end{proof}

\begin{remark}
In \cite{Doeberlein-Schweizer} (see also \cite{Doeberlein-Schweizer-Stricker}) it was shown that under suitable conditions for two equivalent probability measures $\bbr$ and $\hat{\bbr}$ such that $B$ is a savings account implied by $\bbs$ relative to $\bbr$, and $\hat{B}$ is a savings account implied by $\bbs$ relative to $\hat{\bbr}$, we have $B = \hat{B}$ up to an evanescent set.
\end{remark}

\section{Equivalent local martingale deflators}\label{sec-ELMD}

In this section we clarify when a deflator, which is a multiplicative special semimartingale, exists. Recall from Theorem \ref{thm-FTAP} and Remark \ref{rem-more-general} that this ensures that the market satisfies NUPBR for all its nonnegative, self-financing portfolios.

Let $\bbs = \{ S^i : i \in I \}$ be a financial market consisting of nonnegative semimartingales $S^i \geq 0$ with an arbitrary index set $I \neq \emptyset$. We assume that for each $i \in I$ the semimartingale $S^i$ cannot revive from bankruptcy, which means that $S_t^i = 0$ for all $t \geq T_i$, where $T_i$ denotes the \emph{bankruptcy time} of $S^i$ given by
\begin{align*}
T_i := \inf \{ t \in \bbr_+ : S_{t-}^i = 0 \text{ or } S_t^i = 0 \}.
\end{align*}
Then we have
\begin{align}\label{stock-i}
S^i = S_0^i \cale(X^i),
\end{align}
where $X^i$ denotes the semimartingale
\begin{align*}
X^i := (S_-^i)^{-1} \bbI_{\IL 0,T_i \IR} \bdot S^i,
\end{align*}
and where $\cale$ denotes the stochastic exponential; see \cite[Sec. I.4.f]{Jacod-Shiryaev}. Note that $X_0^i = 0$ and $X^i = (X^i)^{T_i}$. Before we proceed, we recall that $\calv$ denotes the space of all c\`{a}dl\`{a}g, adapted processes of locally finite variation starting in zero, and that $\cala_{\loc}$ denotes the space of all locally integrable processes from $\calv$; see \cite[Sec. I.3.a]{Jacod-Shiryaev}. Furthermore, we recall that $\cals$ denotes the space of all semimartingales, and that $\cals_p$ denotes the space of all special semimartingales; see \cite[Sec. I.4.c]{Jacod-Shiryaev}. For each truncation function $h^i \in \calc_t$ (see \cite[Def. II.2.3]{Jacod-Shiryaev}) we have
\begin{align*}
X^i = X^i(h^i) + \breve{X}^i(h^i),
\end{align*}
where $\breve{X}^i(h^i) \in \calv$ and $X^i(h^i) \in \cals_p$ are defined as
\begin{align*}
\breve{X}^i(h^i) &:= \sum_{s \leq \bullet} [ \Delta X_s^i - h^i(\Delta X_s^i) ],
\\ X^i(h^i) &:= X - \breve{X}^i(h^i).
\end{align*}
Here for every optional process $H$ we agree that $\sum_{s \leq \bullet} H_s$ denotes the process given by $\sum_{s \leq t} H_s$ for each $t \in \bbr_+$, provided the series are convergent. We denote by
\begin{align*}
X^i(h^i) = M^i(h^i) + A^i(h^i)
\end{align*}
the canonical decomposition of the special semimartingale $X^i(h^i)$. Then the semimartingale $X^i$ has the decomposition
\begin{align}\label{deco-Si}
X^i = M^i(h^i) + A^i(h^i) + \breve{X}^i(h^i).
\end{align}
Let $Z$ be a multiplicative special semimartingale with multiplicative decomposition $Z = D B^{-1}$ for some $D \in \calm_{\loc}$ with $D_0 = 1$ and a savings account $B$. Then we have $D = \cale(-\Theta)$ for some $\Theta \in \calm_{\loc}$ with $\Theta_0 = 0$ and $\Delta \Theta < 1$, and $B = \cale(R)$ for some predictable process $R \in \calv$ with $\Delta R > -1$. Indeed, these two processes are given by the stochastic logarithms $\Theta = - D_-^{-1} \bdot D$ and $R = B_-^{-1} \bdot B$; see \cite[Sec. II.8.a]{Jacod-Shiryaev}. We call $\Theta$ the \emph{market price of risk} and $R$ the \emph{locally risk-free return} or \emph{virtual short rate} of the savings account $B$. In the subsequent results, all upcoming identities are meant up to an evanescent set. Furthermore, we recall that for each $A \in \cala_{\loc}$ the process $A^p$ denotes its \emph{predictable compensator}; see \cite[Sec. I.3.b]{Jacod-Shiryaev}.

\begin{theorem}\label{thm-ELMD-h}
The following statements are equivalent:
\begin{enumerate}
\item[(i)] $Z$ is an ELMD for $\bbs$.

\item[(ii)] $Z$ is an ELMD for $\bbs \cup \{ B \}$.

\item[(iii)] For each $i \in I$ and each truncation function $h^i \in \calc_t$ we have
\begin{align}\label{drift-A-loc}
&[M^i(h^i) + \breve{X}^i(h^i), \Theta] - \breve{X}^i(h^i) \in \cala_{\loc},
\\ \label{drift} &A^i(h^i) - R = ([M^i(h^i) + \breve{X}^i(h^i), \Theta] - \breve{X}^i(h^i))^p.
\end{align}

\item[(iv)] For each $i \in I$ there exists a truncation function $h^i \in \calc_t$ such that we have (\ref{drift-A-loc}) and (\ref{drift}).
\end{enumerate}
\end{theorem}

We will provide the proof of Theorem \ref{thm-ELMD-h} below. Now, let us assume that for each $i \in I$ the process $X^i$ is a special semimartingale, and consider its canonical decomposition
\begin{align*}
X^i = M^i + A^i.
\end{align*}
Furthermore, let us agree on the notation $\la M,N \ra := [M,N]^p$ for all $M,N \in \calm_{\loc}$ with $[M,N] \in \cala_{\loc}$. Due to \cite[Prop. I.4.50.b]{Jacod-Shiryaev}, this is consistent with the definition of the predictable quadratic covariation. The proof of the following result is similar to that of Theorem \ref{thm-ELMD-h}; indeed, the arguments are even simpler. 

\begin{theorem}\label{thm-ELMD-special}
If $X^i \in \cals_p$ for each $i \in I$, then the following statements are equivalent:
\begin{enumerate}
\item[(i)] $Z$ is an ELMD for $\bbs$.

\item[(ii)] $Z$ is an ELMD for $\bbs \cup \{ B \}$.

\item[(iii)] For each $i \in I$ we have $[M^i,\Theta] \in \cala_{\loc}$ and
\begin{align}\label{drift-final}
A^i - R = \la M^i,\Theta \ra.
\end{align}
\end{enumerate}
\end{theorem}

Now, we prepare the proof of Theorem \ref{thm-ELMD-h}. For this purpose, we provide some auxiliary results. The predictable process $\widetilde{R} \in \calv$ given by
\begin{align*}
\widetilde{R} := \frac{1}{1 + \Delta R} \bdot R
\end{align*}
satisfies $\Delta \widetilde{R} < 1$, and we have $\widetilde{R}^c = R^c$. Since the inverse of $(-1,\infty) \to (-\infty,1)$, $x \mapsto x/(1+x)$ is given by $(-\infty,1) \to (-1,\infty)$, $x \mapsto x/(1-x)$, we have
\begin{align*}
R = \frac{1}{1 - \Delta \widetilde{R}} \bdot \widetilde{R}.
\end{align*}
Therefore, it follows that
\begin{align}\label{R-identities}
\Delta \widetilde{R} = \frac{\Delta R}{1 + \Delta R}, \quad \Delta R = \frac{\Delta \widetilde{R}}{1 - \Delta \widetilde{R}}, \quad (1 + \Delta R) (1 - \Delta \widetilde{R}) = 1, \quad R = \widetilde{R} + [R,\widetilde{R}],
\end{align}
and by Yor's formula (see \cite[II.8.19]{Jacod-Shiryaev}) we obtain
\begin{align*}
B^{-1} = \cale(-\widetilde{R}).
\end{align*}
Now, we define the two processes
\begin{align}\label{theta-tilde}
\widetilde{\Theta} &:= \Theta - [\Theta,\widetilde{R}] = (1 - \Delta \widetilde{R}) \bdot \Theta,
\\ \label{deco-2} Y &:= \widetilde{\Theta} + \widetilde{R}.
\end{align}
Then we have $\widetilde{\Theta} \in \calm_{\loc}$ with $\widetilde{\Theta}_0 = 0$ and $\widetilde{\Theta}^c = \Theta^c$ as well as $Y \in \cals_p$ with $Y_0 = 0$ and $\Delta Y < 1$. Furthermore, we have
\begin{align*}
\Theta = \widetilde{\Theta} + [\widetilde{\Theta},R] = (1 + \Delta R) \bdot \widetilde{\Theta}
\end{align*}
as well as
\begin{align}\label{jumps-Y}
\frac{\Delta Y - 1}{1 - \Delta \widetilde{R}} &= \Delta \Theta - 1.
\end{align}
By Yor's formula (see \cite[II.8.19]{Jacod-Shiryaev}) we obtain
\begin{align}\label{Z-cale-Y}
Z = D B^{-1} = \cale(-\Theta) \cale(-\widetilde{R}) = \cale(-Y).
\end{align}

\begin{lemma}\label{lemma-E-formel}
The following statements are equivalent:
\begin{enumerate}
\item[(i)] $Z$ is an ELMD for $\bbs$.

\item[(ii)] For each $i \in I$ we have $X^i - Y - [X^i,Y] \in \calm_{\loc}$.
\end{enumerate}
\end{lemma}

\begin{proof}
Taking into account (\ref{stock-i}) and (\ref{Z-cale-Y}), this is a consequence of Yor's formula (see \cite[II.8.19]{Jacod-Shiryaev}).
\end{proof}

\begin{proposition}\label{prop-semimartingales-h}
The following statements are equivalent:
\begin{enumerate}
\item[(i)] $Z$ is an ELMD for $\bbs$.

\item[(ii)] $Z$ is an ELMD for $\bbs \cup \{ B \}$.

\item[(iii)] For each $i \in I$ and each truncation function $h^i \in \calc_t$ we have $[X^i,Y] - \breve{X}^i(h^i) \in \cala_{\loc}$ and
\begin{align}\label{drift-R-Theta-h}
A^i(h^i) - \widetilde{R} = ( [ X^i,Y ] - \breve{X}^i(h^i) )^p.
\end{align}

\item[(iv)] For each $i \in I$ there exists a truncation function $h^i \in \calc_t$ such that we have $[X^i,Y] - \breve{X}^i(h^i) \in \cala_{\loc}$ and (\ref{drift-R-Theta-h}).
\end{enumerate}
\end{proposition}

\begin{proof}
(i) $\Rightarrow$ (ii): This implication follows, because $BZ = D \in \calm_{\loc}$.

\noindent(ii) $\Rightarrow$ (iii): Let $i \in I$ and $h^i \in \calc_t$ be arbitrary. By Lemma \ref{lemma-E-formel} we have $X^i - Y - [X^i,Y] \in \calm_{\loc}$. Taking into account the decompositions (\ref{deco-Si}) and (\ref{deco-2}), we obtain
\begin{align*}
M^i(h^i) + A^i(h^i) + \breve{X}^i(h^i) - \widetilde{\Theta} - \widetilde{R} - [X^i,Y] \in \calm_{\loc},
\end{align*}
which implies
\begin{align*}
A^i(h^i) + \breve{X}^i(h^i) - \widetilde{R} - [X^i,Y] \in \calm_{\loc} \cap \calv.
\end{align*}
Taking into account \cite[Lemmas I.3.10, I.3.11]{Jacod-Shiryaev}, we have $[X^i,Y] - \breve{X}^i(h^i) \in \cala_{\loc}$ and
\begin{align}\label{drift-cond-proof-h}
A^i(h^i) - \widetilde{R} - ([X^i,Y] - \breve{X}^i(h^i))^p \in \calm_{\loc} \cap \calv.
\end{align}
Note that the process on the left-hand side of (\ref{drift-cond-proof-h}) is predictable. Hence, according to \cite[Cor. I.3.16]{Jacod-Shiryaev} we obtain (\ref{drift-R-Theta-h}) up to an evanescent set.

\noindent(iii) $\Rightarrow$ (iv): This implication is obvious.

\noindent(iv) $\Rightarrow$ (i): Let $i \in I$ be arbitrary, and let $h^i \in \calc_t$ be a truncation function such that we have $[X^i,Y] - \breve{X}^i(h^i) \in \cala_{\loc}$ and (\ref{drift-R-Theta-h}). Then by the decompositions (\ref{deco-Si}) and (\ref{deco-2}) we have
\begin{align*}
&X^i - Y - [X^i,Y] = M^i(h^i) + A^i(h^i) + \breve{X}^i(h^i) - \widetilde{\Theta} - \widetilde{R} - [X^i,Y]
\\ &= M^i(h^i) + ( [ X^i,Y ] - \breve{X}^i(h^i) )^p - ([X^i,Y] - \breve{X}^i(h^i)) - \widetilde{\Theta} \in \calm_{\loc}.
\end{align*}
Therefore, by Lemma \ref{lemma-E-formel} the process $Z$ is an ELMD for $\bbs$.
\end{proof}

For what follows, we fix an index $i \in I$ and a truncation function $h^i \in \calc_t$. We define the two processes $B^i(h^i), C^i(h^i) \in \calv$ as
\begin{align*}
B^i(h^i) &:= [M^i(h^i),\widetilde{\Theta}] + [\breve{X}^i(h^i),Y] - \breve{X}^i(h^i),
\\ C^i(h^i) &:= [ M^i(h^i) + \breve{X}^i(h^i), \Theta ] - \breve{X}^i(h^i).
\end{align*}

\begin{proposition}\label{prop-h-2}
The following statements are equivalent:
\begin{enumerate}
\item[(i)] We have $[X^i,Y] - \breve{X}^i(h^i) \in \cala_{\loc}$ and (\ref{drift-R-Theta-h}).

\item[(ii)] We have $B^i(h^i) \in \cala_{\loc}$ and
\begin{align}\label{drift-h-2}
A^i(h^i) - \widetilde{R} = [A^i(h^i),\widetilde{R}] + B^i(h^i)^p.
\end{align}
\end{enumerate}
\end{proposition}

\begin{proof}
By the decompositions (\ref{deco-Si}), (\ref{deco-2}) and \cite[Prop. I.4.49.a]{Jacod-Shiryaev} we have
\begin{align*}
[X^i,Y] - \breve{X}^i(h^i) &= [M^i(h^i),\widetilde{\Theta}] + [M^i(h^i),\widetilde{R}] + [A^i(h^i),\widetilde{\Theta}] + [A^i(h^i),\widetilde{R}]
\\ &\quad + [\breve{X}^i(h^i),Y] - \breve{X}^i(h^i)
\\ &= B^i(h^i) + [M^i(h^i),\widetilde{R}] + [A^i(h^i),\widetilde{\Theta}] + [A^i(h^i),\widetilde{R}].
\end{align*}
Furthermore, by \cite[Prop. I.4.49.c]{Jacod-Shiryaev} we have $[M^i(h^i),\widetilde{R}], [A^i(h^i),\widetilde{\Theta}] \in \calm_{\loc}$. Therefore, taking into account \cite[Lemmas I.3.10, I.3.11]{Jacod-Shiryaev} we have $[X^i,Y] - \breve{X}^i(h^i) \in \cala_{\loc}$ if and only if $B^i(h^i) \in \cala_{\loc}$, and in this case, we have
\begin{align*}
([X^i,Y] - \breve{X}^i(h^i))^p = [A^i(h^i),\widetilde{R}] + B^i(h^i)^p,
\end{align*}
completing the proof.
\end{proof}

Before we proceed, recall that each $A \in \calv$ admits a unique decomposition $A = A^c + A^d$ with a continuous process $A^c \in \calv$ and purely discontinuous process $A^d \in \calv$, and that each purely discontinuous $A \in \calv$ admits a unique decomposition $A = A^q + A^a$ with purely discontinuous processes $A^q,A^a \in \calv$ such that $A^q$ is quasi-left-continuous and $A^a$ has only accessible jumps; see \cite[Thm. 4.25]{He}.

\begin{lemma}\label{lemma-B-C-1}
We have $B^i(h^i)^c = C^i(h^i)^c$ and $B^i(h^i)^{dq} = C^i(h^i)^{dq}$.
\end{lemma}

\begin{proof}
By (\ref{theta-tilde}) we have
\begin{align*}
B^i(h^i)^c = [M^i(h^i),\widetilde{\Theta}]^c = (1 - \Delta \widetilde{R}) \bdot [M^i(h^i),\Theta]^c = [M^i(h^i),\Theta]^c = C^i(h^i)^c.
\end{align*}
Now, we define the two purely discontinuous processes $I,J \in \calv$ as
\begin{align*}
I &:= [\breve{X}^i(h^i),Y] - \breve{X}^i(h^i),
\\ J &:= [\breve{X}^i(h^i),\Theta] - \breve{X}^i(h^i).
\end{align*}
Then by \cite[Prop. I.4.49.a]{Jacod-Shiryaev} and (\ref{jumps-Y}) we have
\begin{align*}
I &= (\Delta Y - 1) \bdot \breve{X}^i(h^i),
\\ J &= (\Delta \Theta - 1) \bdot \breve{X}^i(h^i) = \frac{\Delta Y - 1}{1 - \Delta \widetilde{R}} \bdot \breve{X}^i(h^i).
\end{align*}
According to \cite[Thm. 4.21]{He} there are an exhausting sequence $(T_n)_{n \in \bbn}$ of totally inaccessible stopping times and an exhausting sequence $(S_m)_{m \in \bbn}$ of predictable stopping times such that
\begin{align*}
\{ \Delta \breve{X}^i(h^i) \neq 0 \} \subset \bigcup_{n \in \bbn} \IL T_n \IR \cup \bigcup_{m \in \bbn} \IL S_m \IR.
\end{align*}
Since the process $\widetilde{R}$ is predictable, by the construction in the proof of \cite[Thm. 4.25]{He} and \cite[Prop. I.2.24]{Jacod-Shiryaev} we have
\begin{align*}
I^q = \sum_{n \in \bbn} \Delta I_{T_n} \bbI_{\IL T_n,\infty \IL} = \sum_{n \in \bbn} \Delta J_{T_n} \bbI_{\IL T_n,\infty \IL} = J^q.
\end{align*}
Similarly, by (\ref{theta-tilde}) we obtain
\begin{align*}
[M^i(h^i),\widetilde{\Theta}]^{dq} = \big( (1 - \Delta \widetilde{R}) \bdot [M^i(h^i),\Theta]^d \big)^q = [M^i(h^i),\Theta]^{dq}.
\end{align*}
Consequently, we obtain
\begin{align*}
B^i(h^i)^{dq} = [M^i(h^i),\widetilde{\Theta}]^{dq} + I^q = [M^i(h^i),\Theta]^{dq} + J^q = C^i(h^i)^{dq},
\end{align*}
completing the proof.
\end{proof}

\begin{lemma}\label{lemma-B-C-2}
We have
\begin{align}\label{Delta-C}
\Delta C^i(h^i) &= (1 + \Delta R) \Delta B^i(h^i),
\\ \label{Delta-B} \Delta B^i(h^i) &= (1 - \Delta \widetilde{R}) \Delta C^i(h^i).
\end{align}
\end{lemma}

\begin{proof}
By (\ref{R-identities}), (\ref{theta-tilde}) and (\ref{jumps-Y}) we obtain
\begin{align*}
(1 + \Delta R) \Delta B^i(h^i) &= \frac{1}{1 - \Delta \widetilde{R}} \Delta ( [M^i(h^i),\widetilde{\Theta}] + [\breve{X}^i(h^i),Y] - \breve{X}^i(h^i) )
\\ &= \frac{1}{1 - \Delta \widetilde{R}} ( (1 - \Delta \widetilde{R}) \Delta [M^i(h^i),\Theta] + (\Delta Y - 1) \Delta \breve{X}^i(h^i) )
\\ &= \Delta M^i(h^i) \Delta \Theta + (\Delta \Theta - 1) \Delta \breve{X}^i(h^i)
\\ &= \Delta ( [ M^i(h^i) + \breve{X}^i(h^i), \Theta ] - \breve{X}^i(h^i) ) = \Delta C^i(h^i),
\end{align*}
showing (\ref{Delta-C}). Now, the identity (\ref{Delta-B}) follows from (\ref{R-identities}).
\end{proof}

\begin{lemma}\label{lemma-B-C-3}
We have
\begin{align*}
C^i(h^i)^d = B^i(h^i)^d + [R,B^i(h^i)]^d \quad \text{and} \quad B^i(h^i)^d = C^i(h^i)^d - [\widetilde{R},B^i(h^i)]^d.
\end{align*}
\end{lemma}

\begin{proof}
This is an immediate consequence of Lemma \ref{lemma-B-C-2}.
\end{proof}

\begin{proposition}\label{prop-B-C-A-loc}
We have $B^i(h^i) \in \cala_{\loc}$ if and only if $C^i(h^i) \in \cala_{\loc}$.
\end{proposition}

\begin{proof}
Suppose that $B^i(h^i) \in \cala_{\loc}$. By \cite[Lemma I.3.10]{Jacod-Shiryaev} we have
\begin{align*}
\Var [R,B^i(h^i)]^d &= \sum_{s \leq \bullet} | \Delta R_s \Delta B^i(h^i)_s | \leq \Var(R) \sum_{s \leq \bullet} |\Delta B^i(h^i)_s|
\\ &\leq \Var(R) \Var(B^i(h^i)) \in \cala_{\loc}.
\end{align*}
Therefore, by Lemma \ref{lemma-B-C-3} we deduce that $C^i(h^i) \in \cala_{\loc}$. The converse implication is proven analogously.
\end{proof}

\begin{lemma}\label{lemma-B-C-A-loc}
Suppose that $B^i(h^i) \in \cala_{\loc}$, or equivalently $C^i(h^i) \in \cala_{\loc}$. Then we have
\begin{align}\label{comp-part-1}
(B^i(h^i)^p)^c &= (C^i(h^i)^p)^c,
\\ \label{comp-part-2} \Delta C^i(h^i)^p &= (1 + \Delta R) \Delta B^i(h^i)^p,
\\ \label{comp-part-3} \Delta B^i(h^i)^p &= (1 - \Delta \widetilde{R}) \Delta C^i(h^i)^p.
\end{align}
\end{lemma}

\begin{proof}
By \cite[Thm. 7.14]{He} and Lemma \ref{lemma-B-C-1} we have
\begin{align*}
(B^i(h^i)^p)^c = (B^i(h^i)^c + B^i(h^i)^{dq})^p = (C^i(h^i)^c + C^i(h^i)^{dq})^p = (C^i(h^i)^p)^c,
\end{align*}
showing (\ref{comp-part-1}). Furthermore, by \cite[I.3.21, I.2.30]{Jacod-Shiryaev} and Lemma \ref{lemma-B-C-2} we obtain
\begin{align*}
\Delta C^i(h^i)^p &= {}^p [\Delta C^i(h^i)] = {}^p [(1 + \Delta R) \Delta B^i(h^i)] = (1 + \Delta R) \, {}^p [\Delta B^i(h^i)]
\\ &= (1 + \Delta R) \Delta B^i(h^i)^p,
\end{align*}
showing (\ref{comp-part-2}). Now, the identity (\ref{comp-part-3}) is a consequence of (\ref{R-identities}).
\end{proof}

\begin{proposition}\label{prop-h-final}
Suppose that $B^i(h^i) \in \cala_{\loc}$, or equivalently $C^i(h^i) \in \cala_{\loc}$. Then we have (\ref{drift-h-2}) if and only if
\begin{align}\label{drift-h-3}
A^i(h^i) - R = C^i(h^i)^p.
\end{align}
\end{proposition}

\begin{proof}
We have (\ref{drift-h-2}) if and only if
\begin{align}\label{drift-h-2-c}
A^i(h^i)^c - R^c &= (B^i(h^i)^p)^c \quad \text{and}
\\ \label{drift-h-2-d} \Delta A^i(h^i) - \Delta \widetilde{R} &= \Delta [A^i(h^i), \widetilde{R}] + \Delta B^i(h^i)^p.
\end{align}
By virtue of (\ref{R-identities}), condition (\ref{drift-h-2-d}) is equivalent to
\begin{align}\label{drift-h-2-e}
\Delta A^i(h^i) - \Delta R = \frac{\Delta B^i(h^i)^p}{1 - \Delta \widetilde{R}}.
\end{align}
Furthermore, we have (\ref{drift-h-3}) if and only if
\begin{align}\label{drift-h-3-c}
A^i(h^i)^c - R^c &= (C^i(h^i)^p)^c \quad \text{and}
\\ \label{drift-h-3-d} \Delta A^i(h^i) - \Delta R &= \Delta C^i(h^i)^p.
\end{align}
Using Lemma \ref{lemma-B-C-A-loc}, we obtain the equivalences (\ref{drift-h-2-c}) $\Leftrightarrow$ (\ref{drift-h-3-c}) and (\ref{drift-h-2-e}) $\Leftrightarrow$ (\ref{drift-h-3-d}), completing the proof.
\end{proof}

Now, the proof of Theorem \ref{thm-ELMD-h} is a consequence of Propositions \ref{prop-semimartingales-h}, \ref{prop-h-2} and \ref{prop-B-C-A-loc}, \ref{prop-h-final}.

\section{Jump-diffusion models with fixed times of discontinuities}\label{sec-jd}

In this section we study the existence of ELMDs for jump-diffusion models with fixed times of discontinuities. Let $\lambda$ be the Lebesgue measure on $(\bbr_+,\mathcal{B}(\bbr_+))$, and let $W$ be an $\bbr^m$-valued standard Wiener process for some $m \in \bbn$. Furthermore, let $\mu^c$ be a homogeneous Poisson random measure (see \cite[Sec. II.1.c]{Jacod-Shiryaev}) on some mark space $(E,\cale)$, which we assume to be a Blackwell space. Then its predictable compensator (see \cite[Thm. II.1.8]{Jacod-Shiryaev}) is of the form $\nu^c = \lambda \otimes F$ with some $\sigma$-finite measure $F$ on the mark space $(E,\cale)$. Let $\mu^d$ be another Poisson random measure on $(E,\cale)$, and denote by $\nu^d$ its predictable compensator. We set $a_t := \nu^d(\{ t \} \times E)$ for each $t \in \bbr_+$, and define the set $J \subset \bbr_+$ as $J := \{ t \in \bbr_+ : a_t > 0 \}$. According to \cite[Prop. II.1.17]{Jacod-Shiryaev} the set $J$ is countable and we may assume that $a_t \leq 1$ for all $t \in \bbr_+$. We assume that $\mu^d$ is purely discontinuous in the sense that $\nu^d(dt,dx) = \nu^d(dt,dx) \bbI_J(t)$. Let $\zeta^J$ be the measure on $(\bbr_+,\mathcal{B}(\bbr_+))$ given by $\zeta^J(B) = \sum_{k \in B \cap J} 1$ for each $B \in \mathcal{B}(\bbr_+)$; that is, $\zeta^J$ is the counting measure with support $J$.

Let $L_{\loc}^1(\lambda)$ be the space of all predictable processes $\alpha : \Omega \times \bbr_+ \to \bbr$ such that $|\alpha| \bdot \lambda \in \cala_{\loc}$, let $L_{\loc}^1(\zeta^J)$ be the space of all predictable processes $\alpha : \Omega \times \bbr_+ \to \bbr$ such that $|\alpha| \bdot \zeta^J \in \cala_{\loc}$, and let $L_{\loc}^2(W)$ be the space of all predictable processes $\sigma : \Omega \times \bbr_+ \to \bbr^m$ such that $\| \sigma \|_{\bbr^m}^2 \bdot \lambda \in \cala_{\loc}$. Furthermore, let $L_{\loc}^2(\mu^c)$ be the space of all predictable processes $\gamma : \Omega \times \bbr_+ \times E \to \bbr$ such that $|\gamma|^2 * \nu^c \in \cala_{\loc}$, and let $L_{\loc}^2(\mu^d)$ be the space of all predictable processes $\delta : \Omega \times \bbr_+ \times E \to \bbr$ such that $(\delta - \widehat{\delta})^2 * \nu^d + \sum_{s \leq \bullet} (1 - a_s) \widehat{\delta}_s^2 \in \cala_{\loc}$, where
\begin{align*}
\widehat{\delta}_t := \int_E \delta_t(x) \nu^d(\{ t \} \times dx), \quad t \in \bbr_+.
\end{align*}
As in the previous section, we consider a financial market $\bbs = \{ S^i : i \in I \}$ with an arbitrary index set $I \neq \emptyset$. We assume that for each $i \in I$ the semimartingale $S^i$ is given by
\begin{align*}
S^i = S_0^i \cale \big( \alpha^{i,c} \bdot \lambda + \alpha^{i,d} \bdot \zeta^J + \sigma^i \bdot W + \gamma^i * (\mu^c - \nu^c) + \delta^i * (\mu^d - \mu^d) \big)
\end{align*}
with $\alpha^{i,c} \in L_{\loc}^1(\lambda)$, $\alpha^{i,d} \in L_{\loc}^1(\zeta^J)$, $\sigma^i \in L_{\loc}^2(W)$, $\gamma^i \in L_{\loc}^2(\mu^c)$ such that $\gamma^i > -1$, and $\delta^i \in L_{\loc}^2(\mu^d)$ such that $\delta^i - \widehat{\delta}^i > -1$. Here `$\,\, * \,$' denotes the stochastic integral with respect to a random measure; see \cite[Sec. II.1.d]{Jacod-Shiryaev}. Let $Z$ be a multiplicative special semimartingale with multiplicative decomposition $Z = D B^{-1}$, where
\begin{align*}
D = \cale \big( -\theta \bdot W - \psi * (\mu^c - \nu^c) - \phi * (\mu^d - \nu^d) \big) \quad \text{and} \quad B = \cale \big( r^c \bdot \lambda + r^d \bdot \zeta^J \big)
\end{align*}
with $\theta \in L_{\loc}^2(W)$, $\psi \in L_{\loc}^2(\mu^c)$ such that $\psi < 1$, and $\phi \in L_{\loc}^2(\mu^d)$ such that $\phi - \widehat{\phi} < 1$, as well as $r^c \in L_{\loc}^1(\lambda)$ and $r^d \in L_{\loc}^1(\zeta^J)$.

\begin{theorem}\label{thm-jd}
The following statements are equivalent:
\begin{enumerate}
\item[(i)] $Z$ is an ELMD for $\bbs$.

\item[(ii)] $Z$ is an ELMD for $\bbs \cup \{ B \}$.

\item[(iii)] For each $i \in I$ we have
\begin{align}
\alpha^{i,c} - r^c &= \langle \sigma^i,\theta \rangle_{\bbr^m} + \langle \gamma^i, \psi \rangle_{L^2(F)} \quad \text{$\lambda$-a.e.} \quad \text{$\bbp$-a.e.}
\\ \alpha_t^{i,d} - r_t^d &= \int_E \delta_t^i(x) \phi_t(x) \nu^d(\{ t \} \times dx) - \widehat{\delta}_t^i \widehat{\phi}_t, \quad t \in J, \quad \text{$\bbp$-a.e.}
\end{align}
\end{enumerate}
\end{theorem}

Before we provide the proof of Theorem \ref{thm-jd}, let us prepare an auxiliary result.

\begin{lemma}\label{lemma-qv-int-d}
Let $\delta \in L_{\loc}^2(\mu^d)$ be arbitrary, and define the purely discontinuous local martingale $M := \delta * (\mu^d - \nu^d)$. Then the predictable quadratic variation $\la M,M \ra$ is purely discontinuous, and we have
\begin{align*}
\Delta \la M,M \ra_t = \int_E \delta_t(x)^2 \nu^d(\{ t \} \times dx) - \widehat{\delta}_t^2, \quad t \in \bbr_+.
\end{align*}
\end{lemma}

\begin{proof}
According to \cite[Thm. II.1.33.a]{Jacod-Shiryaev} we have
\begin{align*}
\la M,M \ra = (\delta - \widehat{\delta})^2 * \nu^d + \sum_{s \leq \bullet} (1 - a_s) \widehat{\delta}_s^2.
\end{align*}
For each $t \in \bbr_+$ we obtain
\begin{align*}
(\delta - \widehat{\delta})^2 * \nu_t^d &= \int_{[0,t] \times E} (\delta_s(x) - \widehat{\delta}_s)^2 \nu^d(ds,dx) = \sum_{s \leq t} \int_E (\delta_s(x) - \widehat{\delta}_s)^2 \nu^d(\{ s \} \times dx)
\\ &= \sum_{s \leq t} \int_E \big( \delta_s(x)^2 - 2 \delta_s(x) \widehat{\delta}_s + \widehat{\delta}_s^2 \big) \nu^d(\{ s \} \times dx)
\\ &= \sum_{s \leq t} \bigg( \int_E \delta_s(x)^2 \nu^d(\{ s \} \times dx) - 2 \widehat{\delta}_s^2 + \widehat{\delta}_s^2 a_s \bigg),
\end{align*}
and hence
\begin{align*}
\la M,M \ra = \sum_{s \leq \bullet} \bigg( \int_E \delta_s(x)^2 \nu^d(\{ s \} \times dx) - \widehat{\delta}_s^2 \bigg),
\end{align*}
completing the proof.
\end{proof}

For the upcoming proof of Theorem \ref{thm-jd}, we recall that $\calh_{\loc}^2$ denotes the space of all locally square-integrable martingales.

\begin{proof}[Proof of Theorem \ref{thm-jd}]
We define the processes
\begin{align*}
M^i &:= \sigma^i \bdot W + \gamma^i * (\mu^c - \nu^c) + \delta^i * (\mu^d - \nu^d), \quad i \in I,
\\ A^i &:= \alpha^{i,c} \bdot \lambda + \alpha^{i,d} \bdot \zeta^J, \quad i \in I,
\\ \Theta &:= \theta \bdot W + \psi * (\mu^c - \mu^c) + \phi * (\nu^d - \nu^d),
\\ R &:= r^c \bdot \lambda + r^d \bdot \zeta^J.
\end{align*}
According to \cite[Thm. II.1.33.a]{Jacod-Shiryaev} we have $M^i \in \calh_{\loc}^2$ for each $i \in I$, and we have $\Theta \in \calh_{\loc}^2$. Let $i \in I$ be arbitrary. By \cite[Thm. I.4.50.b]{Jacod-Shiryaev} we have $[M^i,\Theta] \in \cala_{\loc}$. Furthermore, by \cite[Thm. II.1.33.a]{Jacod-Shiryaev} we obtain
\begin{align*}
\la M^i,\Theta \ra^c = \big( \langle \sigma^i,\theta \rangle_{\bbr^m} + \langle \gamma^i, \psi \rangle_{L^2(F)} \big) \bdot \lambda,
\end{align*}
and by Lemma \ref{lemma-qv-int-d} we have
\begin{align*}
\Delta \la M^i,\Theta \ra_t = \int_E \delta_t^i(x) \phi_t(x) \nu^d(\{ t \} \times dx) - \widehat{\delta}_t^i \widehat{\phi}_t, \quad t \in \bbr_+.
\end{align*}
Consequently, applying Theorem \ref{thm-ELMD-special} completes the proof.
\end{proof}

We conclude this section with considering the particular situation where the risky assets are given by finitely many diffusion processes. More precisely, consider a continuous market $\bbs = \{ S^1,\ldots,S^d \}$, where
\begin{align*}
S^i = S_0^i \cale(X^i), \quad i=1,\ldots,d,
\end{align*}
and where the $\bbr^d$-valued semimartingale $X = (X^1,\ldots,X^d)$ is an It\^{o} process of the form
\begin{align*}
X = \alpha \bdot \lambda + \sigma \bdot W.
\end{align*}
Here we may regard $\sigma$ as an $\bbr^{d \times m}$-valued process. Let $Z$ be a multiplicative special semimartingale with multiplicative decomposition $Z = D B^{-1}$, where
\begin{align*}
D = \cale(-\theta \bdot W) \quad \text{and} \quad B = \cale(r \bdot \lambda) = \exp(r \bdot \lambda)
\end{align*}
with market price of risk $\theta \in L_{\loc}^2(W)$ and short rate $r \in L_{\loc}^1(\lambda)$. As an immediate consequence of Theorem \ref{thm-jd} we obtain the following result.

\begin{corollary}\label{cor-diffusion-models}
$Z$ is an ELMD for $\bbs$ if and only if
\begin{align}\label{eq-diffusion-1}
\sigma \theta = \alpha - r \bbI_{\bbr^d} \quad \text{$\lambda$-a.e.} \quad \text{$\bbp$-a.e.}
\end{align}
where we agree on the notation $\bbI_{\bbr^d} = (1,\ldots,1) \in \bbr^d$.
\end{corollary}

\begin{remark}
Concerning the construction of an ELMD $Z$, there are two possible approaches:
\begin{itemize}
\item Let us fix a market price of risk $\theta$. Then an appropriate short rate $r$ satisfying (\ref{eq-diffusion-1}) exists if and only if $\alpha - \sigma \theta \in \lin \{ \bbI_{\bbr^d} \}$. In this case, the short rate is unique and given by
\begin{align*}
r \bbI_{\bbr^d} = \alpha - \sigma \theta.
\end{align*}
This confirms the statement about the uniqueness of the savings account; see Proposition \ref{prop-savings-account-unique}.

\item Let us fix a short rate $r$. Then an appropriate market price $\theta$ is a solution of the $\bbr^d$-valued linear equation (\ref{eq-diffusion-1}). Depending on the structure of the matrix $\sigma$, there may be several solutions.
\end{itemize}
\end{remark}

\begin{examples}\label{ex-particular}
Let us consider some particular situations:
\begin{enumerate}
\item Assume $d=m=1$ and $\sigma > 0$. This example includes the well-known Black Scholes model. It highlights the fact that there can be several ELMDs. Indeed, equation (\ref{eq-diffusion-1}) is satisfied if and only if
\begin{align*}
\theta = \frac{\alpha - r}{\sigma} \quad \Longleftrightarrow \quad r = \alpha - \sigma \theta.
\end{align*}

\item Assume $d=2$, $m=1$ and $\alpha^1 \neq \alpha^2$, $\sigma^1 = \sigma^2$. This example shows that an ELMD does not need to exist. Indeed, setting $\sigma := \sigma^1$, equation (\ref{eq-diffusion-1}) is satisfied if and only if
\begin{align*}
\sigma \theta = \alpha^1 - r \quad \text{and} \quad \sigma \theta = \alpha^2 - r,
\end{align*}
which is impossible because $\alpha^1 \neq \alpha^2$. Therefore, by Theorem \ref{thm-FTAP} a savings account $B$ such that $\calp_{\sfi,1}^+(\{ S^1,S^2,B \})$ satisfies NUPBR does not exist.

\item Assume $d=1$, $m=2$ and $\sigma^1,\sigma^2 > 0$. This example shows that the market price of risk does not need to be unique. Indeed, equation (\ref{eq-diffusion-1}) is satisfied if and only if
\begin{align*}
\sigma^1 \theta^1 + \sigma^2 \theta^2 = \alpha - r.
\end{align*}
Hence, for a fixed short rate $r$ there are several solutions for the market price of risk $\theta = (\theta^1,\theta^2)$.
\end{enumerate}
\end{examples}

\section{Further examples}\label{sec-further-examples}

In this section we provide further examples which are related to our main results. The first example deals with the construction of arbitrage free markets by using the real-world pricing formula:

\begin{example}[Real-world pricing]\label{example-real-world-pricing}
We fix a savings account $B$, a local martingale $D > 0$ and define the multiplicative special semimartingale $Z := D B^{-1}$. Let $H^1,\ldots,H^d$ be nonnegative $\calf_T$-measurable contingent claims for some $d \in \bbn$ such that 
\begin{align}\label{H-int}
H^i Z_T \in L^1 \quad \text{for all $i=1,\ldots,d$.} 
\end{align}
We define the market $\bbs = \{ S^1,\ldots,S^d \}$ by the \emph{real-world pricing formula}
\begin{align}\label{real-world-formula}
S_t^i := Z_t^{-1} \bbe [ H^i Z_T | \calf_t ], \quad t \in [0,T]
\end{align}
for all $i=1,\ldots,d$. Then $Z$ is an ELMD for $\bbs$, and by Theorem \ref{thm-FTAP} the set $\calp_{\sfi,1}^+(\bbs \cup \{ B \})$ satisfies NUPBR. If $H^1 = 1$, then the real-world pricing formula reads
\begin{align}\label{real-world-bond}
S_t^1 = Z_t^{-1} \bbe [ Z_T | \calf_t ], \quad t \in [0,T],
\end{align}
and the first primary security account is a zero-coupon bond with maturity date $T$.
\end{example}

We can derive the widely applied risk-neutral pricing formula from the more general real-world pricing formula as a special case:

\begin{example}[Risk-neutral pricing]
Suppose that in the setting of Example \ref{example-real-world-pricing} the savings accounts $B$ and $B^{-1}$ are bounded. By Theorem \ref{thm-FTAP-classical} the set $\calp_{\sfi,0}^{\adm}(\bbs \cup \{ B \})$ satisfies NFLVR if and only if $D$ is a true martingale. Suppose that this is the case. Then, also by Theorem \ref{thm-FTAP-classical}, there exists an ELMM $\bbq \approx \bbp$ for $\bbs B^{-1}$, and its density process is given by $D / D_0$. By (\ref{H-int}) we have
\begin{align*}
H^i B_T^{-1} \in L^1(\bbq) \quad \text{for all $i=1,\ldots,d$.} 
\end{align*}
Hence, using the Bayes' Rule (see \cite[III.3.9]{Jacod-Shiryaev}) we obtain the \emph{risk-neutral pricing formula}
\begin{align}\label{risk-neutral-formula}
S_t^i = B_t \, \bbe_{\bbq} [ H^i B_T^{-1} | \calf_t ], \quad t \in [0,T]
\end{align}
for all $i=1,\ldots,d$. If $H^1 = 1$, then the real-world pricing formula reads
\begin{align}\label{risk-neutral-bond}
S_t^1 = B_t \, \bbe_{\bbq} [ B_T^{-1} | \calf_t ], \quad t \in [0,T],
\end{align}
and the first primary security account is a zero-coupon bond with maturity date $T$.
\end{example}

The real-world pricing formula (\ref{real-world-formula}) can also be used for pricing and hedging contingent claims:

\begin{example}[Pricing and hedging of contingent claims]\label{example-pricing}
Let $\bbs = \{ S^1,\ldots,S^d \}$ be a market, and let $Z$ be an ELMD for $\bbs$ which is a multiplicative special semimartingale of the form $Z = D B^{-1}$ with a local martingale $D > 0$ and a savings account $B$. By Theorem \ref{thm-FTAP} the market is free of arbitrage in the sense that $\calp_{\sfi,1}^+(\bbs \cup \{ B \})$ satisfies NUPBR. Now, let $H$ be a nonnegative $\calf_T$-measurable contingent claim such that $H Z_T \in L^1$.
\begin{itemize}
\item We would like to find a consistent price process $\pi$ for the contingent claim $H$, which means that $\calp_{\sfi,1}^+(\bbs \cup \{ \pi \} \cup \{ B \})$ should also satisfy NUPBR. This is achieved by using the real-world pricing formula
\begin{align}\label{real-world-pi}
\pi_t := Z_t^{-1} \bbe [ H Z_T | \calf_t ], \quad t \in [0,T].
\end{align}
Indeed, the process $Z$ is also an ELMD for $\bbs \cup \{ \pi \}$, and hence, by Theorem \ref{thm-FTAP} the set $\calp_{\sfi,1}^+(\bbs \cup \{ \pi \} \cup \{ B \})$ satisfies NUPBR.

\item We assume that a market participant always \emph{prefers more for less}. In this spirit, the price process $\pi$ is the most economical price process because
\begin{align}\label{cheapest}
\pi \leq V^{\nu}
\end{align}
for every self-financing strategy $\nu = (\delta,\eta) \in \Delta_{\sfi}(\bbs \cup \{ B \})$ such that $V^{\nu} \geq 0$ and $V_T^{\nu} = H$, where $V^{\nu}$ denotes the corresponding self-financing portfolio $V^{\nu} = \sum_{i=1}^d \delta^i S^i + \eta B$. In other words, the price process $\pi$ is a lower bound for the least expensive nonnegative self-financing portfolio which replicates $H$, provided such a portfolio exists. In particular, if $\pi$ can be realized as a self-financing portfolio, then it is the least expensive portfolio replicating $H$. In order to show (\ref{cheapest}), note that by Proposition \ref{prop-sf} the process $V^{\nu} Z$ is a nonnegative local martingale, and hence a supermartingale. Furthermore, the process $\pi Z$ is a martingale. Since a martingale is the minimal nonnegative supermartingale that reaches at time $T$ a given nonnegative integrable value (see \cite[Lemma 10.4.1]{Platen}), we deduce $\pi_T Z_T \leq V_T^{\nu} Z_T$, and hence the inequality (\ref{cheapest}) follows.
\end{itemize}
\end{example}

In principle, the real-world prices of contingent claims can also be determined with the risk-neutral pricing formula, but only if the local martingale $D$ is a true martingale. An example of a local martingale, which is not a true martingale, is the inverse of a squared Bessel process of dimension $4$, which is given by the solution of the stochastic differential equation
\begin{align}\label{Bessel}
d D_t = - 2 D_t^{\frac{3}{2}} dW_t,
\end{align}
where $W$ is a real-valued standard Wiener process. This process plays an important role in the minimal market model; see, for example \cite[Sec. 13]{Platen}. The following example shows that the real-world pricing formula also provides the forward measure pricing formula in a market with a zero-coupon bond.

\begin{example}[Forward measure pricing formula]
Let $\bbs = \{ P^T \}$ be a market consisting of a zero-coupon bond $P^T$ with maturity $T$. Let $Z$ be an ELMD for $\bbs$ which is a multiplicative special semimartingale of the form $Z = D B^{-1}$ with a local martingale $D > 0$ and a savings account $B$. By Theorem \ref{thm-FTAP} the set $\calp_{\sfi,1}^+(\{ P^T,B \})$ satisfies NUPBR. Now, let $H$ be a nonnegative $\calf_T$-measurable contingent claim such that $H Z_T \in L^1$. By Example \ref{example-pricing} we know that the least expensive price process $\pi$ is given by the real-world pricing formula (\ref{real-world-pi}). Since $Z$ is an ELMD for $\bbs$, the process $P^T Z$ is a local martingale. If $P^T Z$ is a true martingale, then using the Bayes' Rule (see \cite[III.3.9]{Jacod-Shiryaev}) the price process $\pi$ can be expressed by the \emph{forward measure pricing formula}
\begin{align}\label{formula-forward}
\pi_t = P_t^T \bbe_{\bbq^T} [ H | \calf_t ], \quad t \in [0,T],
\end{align}
where $\bbq^T \approx \bbp$ denotes the $T$-forward measure with density process $(P^T Z) / (P_0^T Z_0)$. For practitioners it may be interesting to note that the convenient forward measure pricing formula (\ref{formula-forward}) can still be applied when $D$ is not a true martingale, which means that a risk-neutral measure $\bbq \approx \bbp$ for the discounted bond $P^T / B$ might not exist; see also Example \ref{example-bond-free-lunch} below. If $D$ is a true martingale, then the measure change $\bbq^T \approx \bbp$ can be performed in the two steps $\bbq \approx \bbp$ and $\bbq^T \approx \bbq$ with respective density processes $D / D_0$ and $(P^T B^{-1}) / P_0^T$, which is usually done under the risk-neutral approach.
\end{example}

We continue the above example by calculating the least expensive zero-coupon bond price $P^T_0$ under the minimal market model with $D$ as in (\ref{Bessel}).

\begin{example}[A free lunch with vanishing risk]\label{example-bond-free-lunch}
As in the previous example, we consider a market $\bbs = \{ P^T \}$ consisting of a zero-coupon bond $P^T$ with maturity $T$. Here we fix a deterministic savings account $B$; for example it could be $B_t = e^{rt}$, $t \in [0,T]$ for some constant interest rate $r$. Furthermore, we denote by $D > 0$ the strict local martingale given by the stochastic differential equation (\ref{Bessel}) with $D_0 = 1$, and we specify the zero-coupon bond $P^T$ by the real-world pricing formula (\ref{real-world-bond}) with $Z = D B^{-1}$, which becomes here
\begin{align*}
P_t^T = \big( B_t B_T^{-1} \big) D_t^{-1} \bbe [ D_T | \calf_t ], \quad t \in [0,T].
\end{align*}
Then $Z$ is an ELMD for $\bbs$, and hence, by Theorem \ref{thm-FTAP} the set $\calp_{\sfi,1}^+(\{ P^T,B \})$ satisfies NUPBR. However, since $D$ is not a true martingale, we cannot apply Theorem \ref{thm-FTAP-classical}, and, as a consequence, we cannot ensure that the set $\calp_{\sfi,0}^{\adm}(\{ P^T,B \})$ satisfies NFLVR. However, the process $P^T Z$ is a true martingale, which allows us to use the convenient forward measure pricing formula (\ref{formula-forward}) for pricing contingent claims $H$. In order to investigate our observation about the possible existence of a free lunch with vanishing risk, let us denote by $\widetilde{P}^T$ the price of a zero-coupon bond with maturity $T$ computed formally with the risk-neutral pricing formula (\ref{risk-neutral-bond}) under some putative risk-neutral probability measure $\bbq$. Since $B$ is deterministic, we obtain
\begin{align*}
\widetilde{P}^T = B / B_T
\end{align*}
for every choice of the measure $\bbq$, and hence $P_0^T < \widetilde{P}_0^T$ because by formula (8.7.17) in \cite{Platen} we have
\begin{align*}
P_0^T = \frac{1}{B_T} \bigg[ 1 - \exp \bigg( - \frac{1}{2T} \bigg) \bigg],
\end{align*}
which is less than the risk-neutral zero-coupon bond $\widetilde{P}_0^T = B_T^{-1}$. In accordance with our previous observation, the difference between these two zero-coupon bond prices indicates the presence of some free lunch with vanishing risk. However, as we have seen, the NUPBR condition is still satisfied. By making $1/Z$ tradeable as num\'{e}raire portfolio, see \cite{Platen} and \cite{Karatzas-Kardaras}, one can hedge the least expensive zero-coupon bond. This allows us to exploit the identified free lunch with vanishing risk, which is caused by the strict local martingale property of D. From the practical perspective, this allows us to produce long-term zero-coupon bond type payouts less expensively than under the NFLVR condition, which is likely to have significant impact on the pension and insurance industry, see e.g. \cite{SZP}, through the realistic modeling of the ELMD $Z$ for $\bbs$, which is the inverse of the num\'{e}raire portfolio $1/Z$. The latter coincides, up to some subtleties, with the growth optimal portfolio; see \cite{Karatzas-Kardaras}.
\end{example}

\section{Dynamic trading strategies}\label{sec-dynamic}

For risk management purposes, trading strategies that do not need to be self-financing, can be crucial; see, e.g., \cite{Schweizer-1991} and \cite{Du-Platen}. In this section we will construct such strategies which are in line with our findings from the previous sections. Of course, when looking for trading strategies which are not self-financing, not all strategies can be allowed. In order to construct reasonable strategies that do not need to be self-financing, the concept of a locally real-world mean self-financing dynamic trading strategy (see \cite{Du-Platen}) turns out to be fruitful. In the risk-neutral context, an analogous notion has been introduced in \cite{Schweizer-1991}.

As in the previous sections, we consider a finite market $\bbs = \{ S^1,\ldots,S^d \}$ with nonnegative semimartingales. We assume that there exists an ELMD $Z$ for $\bbs$, which is a multiplicative special semimartingale of the form $Z = D B^{-1}$ with a local martingale $D > 0$ and a savings account $B$. According to Theorem \ref{thm-FTAP}, the set $\calp_{\sfi,1}^+(\bbs \cup \{ B \})$ satisfies NUPBR.

\begin{definition}
A \emph{dynamic trading strategy} $\nu = (\delta,\eta)$ consists of a self-financing strategy $\delta \in \Delta_{\sfi}(\bbs)$ and a real-valued optional process $\eta$, which is integrable with respect to $B$, such that the portfolio
\begin{align*}
V^{\nu} := S^{\delta} + B^{\eta},
\end{align*}
where we use the common notations $S^{\delta} := \delta \cdot S$ and $B^{\eta} := \eta \cdot B$, satisfies
\begin{align*}
V^{\nu} = S_0^{\delta} + \delta \bdot S + B^{\eta}.
\end{align*}
In this case, we call $\delta$ the \emph{self-financing part} of $\nu$.
\end{definition}

\begin{definition}
Let $\nu = (\delta,\eta)$ be a dynamic trading strategy. The \emph{profit and loss (P\&L) process} $C^{\nu}$ is defined as
\begin{align*}
C^{\nu} := B^{\eta} - B_0^{\eta} - \eta \bdot B.
\end{align*}
\end{definition}

\begin{remark}
Note that $C_0^{\nu} = 0$ and
\begin{align}\label{decomp-V-nu}
V^{\nu} = V_0^{\nu} + \delta \bdot S + \eta \bdot B + C^{\nu}.
\end{align}
Thus, the P\&L process $C^{\nu}$ monitors the cumulative inflow and outflow of extra capital. Note that $C^{\nu} = 0$ if and only if
\begin{align}\label{nu-self-fin}
V^{\nu} = V_0^{\nu} + \delta \bdot S + \eta \bdot B.
\end{align}
This is in particular satisfied if $\nu$ is self-financing in the market $\bbs \cup \{ B \}$, but condition (\ref{nu-self-fin}) is a bit more general, because $\eta$ may be optional, which allows us to go beyond predictable processes. More precisely, we have $\nu \in \Delta_{\sfi}(\bbs \cup \{ B \})$ if and only if $\nu \in L((S,B))$ and equation (\ref{nu-self-fin}) is fulfilled.
\end{remark}

\begin{definition}
A dynamic trading strategy $\nu$ and the corresponding portfolio $V^{\nu}$ are called \emph{locally real-world mean self-financing} if $C^{\nu} \in \calm_{\loc}$.
\end{definition}

Note that every self-financing strategy $\nu$ is locally real-world mean self-financing. The following result shows how we can easily construct locally real-world mean self-financing strategies.

\begin{proposition}\label{prop-rwsf}
Let $\nu = (\delta,\eta)$ be a dynamic trading strategy such that $\eta$ is a local martingale. Then the strategy $\nu$ is locally real-world mean self-financing, and the P\&L process is given by $C^{\nu} = B_- \bdot \eta$.
\end{proposition}

\begin{proof}
According to \cite[Prop. I.4.49.a]{Jacod-Shiryaev} we have
\begin{align*}
B^{\eta} = B_0^{\eta} + B_- \bdot \eta + \eta \bdot B, 
\end{align*}
and hence $C^{\nu} = B_- \bdot \eta \in \calm_{\loc}$, completing the proof.
\end{proof}

Recall that $Z$ denotes an ELMD for the financial market $\bbs$. For every self-financing strategy $\nu$ the corresponding self-financing portfolio $V^{\nu}$ has the property that $V^{\nu} Z$ is also a local martingale, which by virtue of Theorem \ref{thm-FTAP} means that $\calp_{\sfi,1}^+(\bbs \cup \{ V^{\nu} \} \cup \{ B \} )$ also satisfies NUPBR. In this spirit, we wish to construct such dynamic trading strategies $\nu$ such that $V^{\nu} Z$ is also a local martingale. As we will see, this can be achieved by choosing locally real-world mean self-financing strategies which are orthogonal to the deflator $Z$. More precisely, we have the following result.

\begin{theorem}\label{thm-mean-sf}
Let $\nu = (\delta,\eta)$ be a dynamic trading strategy with $V^{\nu} \geq 0$ such that $\eta$ is a continuous local martingale with $[\eta,Z] = 0$. Then the following statements are true:
\begin{enumerate}
\item The strategy $\nu$ is locally real-world mean self-financing.

\item The process $V^{\nu} Z$ is a local martingale.
\end{enumerate}
\end{theorem}

\begin{proof}
By Proposition \ref{prop-rwsf} the strategy $\nu$ is locally real-world mean self-financing, and the P\&L process is given by $C^{\nu} = B_- \bdot \eta$. Noting that $[C^{\nu},Z] = B_- \bdot [\eta,Z] = 0$, by the representation (\ref{decomp-V-nu}) and Proposition \ref{prop-deflator-prop-2} we have $V^{\nu} Z \in \calm_{\sigma}$. Since $V^{\nu} \geq 0$, by Lemma \ref{lemma-M-sigma-loc} we deduce that $V^{\nu} Z \in \calm_{\loc}$.
\end{proof}

Recall that $Z$ is an ELMD for $\calp_{\sfi,1}^+(\bbs \cup \{ B \} \})$. Now, let us denote by $\calp_{{\rm msf},1}^+(\bbs \cup \{ B \})$ the set of all outcomes of trading strategies with initial value one which arise from a nonnegative self-financing portfolio $V^{\nu}$ or from a nonnegative locally real-world mean self-financing portfolio $V^{\nu}$ as in Theorem \ref{thm-mean-sf}. Obviously, the set $\calp_{{\rm msf},1}^+(\bbs \cup \{ B \})$ is larger than $\calp_{{\rm sf},1}^+(\bbs \cup \{ B \})$. However, using Theorem \ref{thm-mean-sf} we deduce that $Z$ is even an ELMD for $\calp_{{\rm msf},1}^+(\bbs \cup \{ B \})$.

\begin{example}[Construction of locally real-world mean self-financing strategies]\label{example-msf-1}
Suppose we are in the situation of Corollary \ref{cor-diffusion-models}, where we consider diffusion models. Then we have $Z = \cale(Y)$, where
\begin{align*}
Y = -\theta \bdot W - r \bdot \lambda
\end{align*}
with an $\bbr^m$-valued standard Wiener process process $W$, a market price of risk $\theta \in L_{\loc}^2(W)$ and a short rate $r \in L_{\loc}^1(\lambda)$. Let $\delta \in \Delta_{\sfi}(\bbs)$ be a self-financing strategy, and let $\vartheta \in L_{\loc}^2(W)$ be such that $V^{\nu} \geq 0$, where $\nu = (\delta,\eta)$ denotes the dynamic trading strategy given by $\eta := \vartheta \bdot W$. We assume $\theta \perp \vartheta$ in the sense that $\theta \cdot \vartheta = 0$. Since $Z = 1 + Z_- \bdot Y$, we obtain
\begin{align*}
[\eta,Z] = [\eta, Z_- \bdot Y] = - \la \vartheta \bdot W, Z_- \theta \bdot W \ra = - Z_- (\vartheta \cdot \theta) \bdot \lambda = 0.
\end{align*}
Therefore, by Theorem \ref{thm-mean-sf} the strategy $\nu$ is locally real-world mean self-financing, and the process $V^{\nu} Z$ is a local martingale.
\end{example}

\begin{example}[Pricing and hedging of contingent claims]\label{example-msf-2}
Consider the framework of Example \ref{example-pricing}, where we have constructed the least expensive price process $\pi$ for a nonnegative contingent claim $H$ with maturity $T$ by using the real-world pricing formula (\ref{real-world-pi}). Note that the claim $H$ can be non-replicable; that is, it can happen that a self-financing portfolio $\nu = (\delta,\eta) \in \Delta_{\sfi}(\bbs \cup \{ B \})$ with $\pi = V^{\nu}$ does not exist. However, allowing dynamic trading strategies that do not need to be self-financing, it might be possible to find a locally real-world mean self-financing strategy $\nu = (\delta,\eta)$ as in Theorem \ref{thm-mean-sf} such that $\pi = V^{\nu}$. Then, by the decomposition (\ref{decomp-V-nu}) the P\&L process $C^{\nu}$ can also be regarded as the non-hedgeable part of the claim $H$. In any case, we can improve the lower bound (\ref{cheapest}) for the least expensive nonnegative portfolio which replicates $H$. Indeed, using our previous results, we even have
\begin{align*}
\pi \leq V^{\nu},
\end{align*}
where $\nu$ is any self-financing strategy or any locally real-world mean self-financing strategy as in Theorem \ref{thm-mean-sf} such that $V^{\nu} \geq 0$ and $V_T^{\nu} = H$. 
\end{example}

\section{The semimartingale property of the primary security accounts}\label{sec-semimartingale}

So far, we have assumed that the financial market consists of nonnegative \emph{semimartingales}. Using the results from \cite{Kardaras-Platen}, we will show in this section that the primary security accounts must be semimartingales under the mild and natural condition that the market satisfies NUPBR for all its self-financing portfolios which are simple and without short positions. We also refer to \cite{Balint-Schweizer-2} for similar results.

Let $\bbs = \{ S^1,\ldots,S^d \}$ be a market consisting of nonnegative, adapted, c\`{a}dl\`{a}g processes $S^1,\ldots,S^d$. Furthermore, let $B$ be a savings account. For what follows, we recall that `$\,\, \bdot \,$' denotes stochastic integration, whereas `$\,\, \cdot \,$' denotes the usual inner product in Euclidean space. A strategy $\nu = (\delta,\eta)$ for $\bbs \cup \{ B \}$ is called \emph{simple} if:
\begin{enumerate}
\item $\delta$ is of the form
\begin{align}\label{simple-1}
\delta = \sum_{j=1}^n \Delta_j \bbI_{\IR \tau_{j-1},\tau_j \IR} 
\end{align}
for some $n \in \bbn$, where $0 = \tau_0 < \tau_1 < \ldots < \tau_n$ are finite stopping times, and $\Delta_j = (\Delta_j^i)_{i=1,\ldots,d}$ is $\calf_{\tau_{j-1}}$-measurable for each $j=1,\ldots,n$, and

\item $\eta$ is a real-valued optional process, which is integrable with respect to $B$.
\end{enumerate}
For such a simple strategy we define the portfolio $V^{\nu} := S^{\delta} + B^{\eta}$, where we use the common notations $S^{\delta} := \delta \cdot S$ and $B^{\eta} := \eta \cdot B$. The portfolio $V^{\nu}$ is called \emph{self-financing} if
\begin{align*}
V^{\nu} = V_0^{\nu} + \sum_{j=1}^n \Delta_j \cdot (S^{\tau_j} - S^{\tau_{j-1}}) + \eta \bdot B,
\end{align*}
where $\tau_0,\ldots,\tau_n$ and $\Delta_1,\ldots,\Delta_n$ stem from the representation (\ref{simple-1}). We denote by $\calp_{\sfi,1,s}^{\nu \geq 0}(\bbs \cup \{ B \})$ the set of all outcomes of simple self-financing portfolios with initial value one such that $\delta^1,\ldots,\delta^d \geq 0$ and $\eta \geq 0$. The latter condition means that no short selling is allowed. We say that a nonnegative semimartingale $S \geq 0$ cannot revive from bankruptcy if $S = 0$ on $\IL \tau,\infty \IL$, where $\tau := \inf \{ t \in \bbr_+ : X_{t-} = 0 \text{ or } X_t = 0 \}$.

\begin{theorem}\label{thm-semimartingales}
Suppose there is a savings account $B$ such that $\calp_{\sfi,1,s}^{\nu \geq 0}(\bbs \cup \{ B \})$ satisfies NUPBR. Then $S^1,\ldots,S^d$ are semimartingales which cannot revive from bankruptcy.
\end{theorem}

For the proof of Theorem \ref{thm-semimartingales} we prepare some auxiliary results. Let $B$ be a savings account. We consider the discounted market $\bbx := \bbs B^{-1}$. A strategy $\theta$ for $\bbx$ is called \emph{simple} if it is of the form
\begin{align}\label{simple-2}
\theta = \sum_{j=1}^n \vartheta_j \bbI_{\IL \tau_{j-1},\tau_j \IL}
\end{align}
for some $n \in \bbn$, where $0 = \tau_0 < \tau_1 < \ldots < \tau_n$ are finite stopping times, and $\vartheta_j = (\vartheta_j^i)_{i=1,\ldots,d}$ is $\calf_{\tau_{j-1}}$-measurable for each $j=1,\ldots,n$. For $x \in \bbr$ and such a simple strategy $\theta$ we define the integral process
\begin{align*}
X^{x,\theta} := x + \sum_{j=1}^n \vartheta_j \cdot (S^{\tau_j} - S^{\tau_{j-1}}),
\end{align*}
where $\tau_0,\ldots,\tau_n$ and $\vartheta_1,\ldots,\vartheta_n$ stem from the representation (\ref{simple-2}). We denote by $\cali_{1,s}^{\theta \geq 0}(\bbx)$ the set of all outcomes of integral processes with $x=1$ and a simple strategy $\theta$ such that $\theta \geq 0$ and $X_-^{x,\theta} - \theta \cdot X_- \geq 0$. Furthermore, let $\Delta_s(\bbx)$ be the set of all simple strategies for $\bbx$, and let $\Delta_{\sfi,s}(\bar{\bbx})$ be the set of all simple self-financing strategies for $\bar{\bbx} := \bbx \cup \{ 1 \}$. Then we have the following result, which is similar to Lemma \ref{lemma-sf-bijection}. 

\begin{lemma}\label{lemma-change-num}
Suppose that $1 \notin \bbx$. Then there is a bijection between $\bbr \times \Delta_s(\bbx)$ and $\Delta_{\sfi,s}(\bar{\bbx})$, which is defined as follows:
\begin{enumerate}
\item For $\delta \in \Delta_{\sfi}(\bar{\bbx})$ we assign
\begin{align}\label{trans-1}
(\delta,\eta) = \nu \mapsto (x,\theta) := (X_0^{\delta},\delta) \in \bbr \times \Delta(\bbx).
\end{align}
\item For $(x,\theta) \in \bbr \times \Delta(\bbx)$ we assign
\begin{align}\label{trans-2}
(x,\theta) \mapsto (\delta,\eta) = \nu = (\theta, X_-^{x,\theta} - \theta \cdot X_- ) \in \Delta_{\sfi}(\bar{\bbx}).
\end{align}
\end{enumerate}
Furthermore, for all $(x,\theta) \in \bbr \times \Delta_s(\bbx)$ and the corresponding self-financing strategy $\nu \in \Delta_{\sfi,s}(\bar{\bbx})$ we have
\begin{align}\label{trans-3}
\bar{X}^{\nu} = X^{x,\theta}.
\end{align}
\end{lemma}

\begin{proposition}\label{prop-semimartingales}
$\calp_{\sfi,1,s}^{\nu \geq 0}(\bbs \cup \{ B \})$ satisfies NUPBR if and only if $\cali_{1,s}^{\theta \geq 0}(\bbs B^{-1})$ satisfies NUPBR.
\end{proposition}

\begin{proof}
The proof is analogous to that of Proposition \ref{prop-NA-concepts-equiv}, where we take into account relation (\ref{trans-2}) from Lemma \ref{lemma-change-num}.
\end{proof}

Now, we are ready to provide the proof of Theorem \ref{thm-semimartingales}.

\begin{proof}[Proof of Theorem \ref{thm-semimartingales}]
By Proposition \ref{prop-semimartingales} the set $\cali_{1,s}^{\theta \geq 0}(\bbs B^{-1})$ satisfies NUPBR. Hence by \cite[Prop. 1.1 and Thm. 1.3]{Kardaras-Platen} the processes $S^1 B^{-1}, \ldots, S^d B^{-1}$ are semimartingales which cannot revive from bankruptcy. Since $B$ is a savings account, it follows that $S^1,\ldots,S^d$ are semimartingales which cannot revive from bankruptcy.
\end{proof}

\section{Filtration enlargements}\label{sec-filtration}

It can happen that some additional information arises, which is not originally present in the market. Models with insider information have been widely studied in the literature; see, for example \cite{AFK, Aksamit-2, Aksamit-3, Aksamit-book, Fontana-Jeanblanc-Song}. Mathematically, such additional information means that we consider an enlargement of the original filtration. More precisely, we consider a new filtration $\bbg = (\calg_t)_{t \in \bbr_+}$ such that $\calf_t \subset \calg_t$ for all $t \in \bbr_+$, where $\bbf = (\calf_t)_{t \in \bbr_+}$ denotes the original filtration. Typically, there is also a $\bbg$-stopping time $\tau$ involved, and it arises the question when absence of arbitrage of a financial market $\bbs = \{ S^1,\ldots,S^d \}$ under the original filtration $\bbf$ implies absence of arbitrage of the stopped market $\bbs^{\tau} = \{ S^{1,\tau}, \ldots, S^{d,\tau} \}$ under the enlarged filtration $\bbg$. There exist several results for the case that the financial market is already discounted by some num\'{e}raire; see, for example, the aforementioned articles. By virtue of Proposition \ref{prop-NA-concepts-equiv}, we can transfer most of these results to the situation which we consider in this paper.

As an illustration, consider the situation with a progressive filtration enlargement, which we briefly recall; see \cite[Sec. 1.3]{AFK} for further details. Let $\tau : \Omega \to [0,\infty]$ be a $\calf$-measurable random time such that $\bbp(\tau = \infty) = 0$. The progressively enlarged filtration $\bbg = (\calg_t)_{t \in \bbr_+}$ is defined as
\begin{align*}
\calg_t := \{ B \in \calf : B \cap \{ \tau > t \} = B_t \cap \{ \tau > t \} \text{ for some } B_t \in \calf_t \} \quad \text{for all $t \in \bbr_+$.}
\end{align*}
Let $Z$ be the Az\'{e}ma supermartingale given by $Z_t = \bbp(\tau > t | \calf_t)$ for all $t \in \bbr_+$, and let $A$ be the dual optional projection of $\bbI_{\IL \tau,\infty \IL}$. Furthermore, we define
\begin{align*}
\zeta := \inf \{ t \in \bbr_+ : Z_t = 0 \},
\end{align*}
the $\calf_{\zeta}$-measurable event $\Lambda := \{ \tau < \infty, Z_{\zeta-} > 0, \Delta A_{\zeta} = 0 \}$ as well as
\begin{align*}
\eta := \zeta_{\Lambda} := \zeta \bbI_{\Lambda} + \infty \bbI_{\Omega \setminus \Lambda}.
\end{align*}
The following result applies in the situation of Theorem \ref{thm-FTAP}.

\begin{theorem}
Suppose there is a savings account $B$ such that $\calp_{\sfi,1}^+(\bbs \cup \{ B \})$ satisfies NUPBR under $\bbf$. If $\bbp(\eta < \infty, \Delta S_{\eta} \neq 0, \Delta B_{\eta} \neq 0) = 0$, where $S = (S^1,\ldots,S^d)$, then $\calp_{\sfi,1}^+(\bbs^{\tau} \cup \{ B^{\tau} \})$ satisfies NUPBR under $\bbg$.
\end{theorem}

\begin{proof}
By Proposition \ref{prop-NA-concepts-equiv} the set $\cali_1^+(\bbs B^{-1})$ satisfies NUPBR under $\bbf$. Note that the discounted market is given by $\bbs B^{-1} = \{ \tilde{S}^1,\ldots,\tilde{S}^d \}$, where $\tilde{S}^i := S^i B^{-1}$ for all $i=1,\ldots,d$. By assumption we have $\bbp(\tau < \infty, \Delta \tilde{S}_{\eta} \neq 0) = 0$. Thus, by \cite[Thm. 1.4]{AFK} the set $\cali_1^+((\bbs B^{-1})^{\tau})$ satisfies NUPBR under $\bbg$, where $(\bbs B^{-1})^{\tau} = \{ \tilde{S}^{1,\tau},\ldots,\tilde{S}^{d,\tau} \}$ denotes the stopped discounted market. Consequently, by Proposition \ref{prop-NA-concepts-equiv} the set $\calp_{\sfi,1}^+(\bbs^{\tau} \cup \{ B^{\tau} \})$ satisfies NUPBR under $\bbg$.
\end{proof}

\section{Financial models in discrete time}\label{sec-discrete}

Using our previous results for continuous time models, we can also derive a no-arbitrage result for discrete time models. This result is in accordance with the well-known result concerning the absence of arbitrage in discrete time finance. In this section, we assume that a discrete filtration $(\calf_k)_{k=0,\ldots,T}$ for some integer $T \in \bbn$ with $\calf_0 = \{ \Omega,\emptyset \}$ is given, and we consider a finite market $\bbs = \{ S^1,\ldots,S^d \}$ consisting of nonnegative, adapted processes. As shown in \cite[page 14]{Jacod-Shiryaev}, this setting can be regarded as a particular case of the continuous time framework, which we have considered so far. Note that every $\bbr^d$-valued predictable process $\delta$ belongs to $\Delta(\bbs)$, and that the stochastic integral $\delta \bdot S = (\delta \bdot S_t)_{t=0,\ldots,T}$ is given by
\begin{align*}
\delta \bdot S_0 &= 0,
\\ \delta \bdot S_t &= \sum_{k=1}^t \delta_k \cdot (S_{k} - S_{k-1}), \quad t=1,\ldots,T.
\end{align*}

\begin{theorem}\label{thm-diskret-2}
The following statements are equivalent:
\begin{enumerate}
\item[(i)] There exists a savings account $B$ such that $\calp_{\sfi,0}(\bbs \cup \{ B \})$ satisfies NA.

\item[(ii)] There exist a savings account $B$ and an EMM $\bbq \approx \bbp$ for $\bbs B^{-1}$.

\item[(iii)] There exists an EMD $Z$ for $\bbs$ which is a multiplicative special semimartingale such that the local martingale part is a true martingale.
\end{enumerate}
If the previous conditions are fulfilled, then the savings accounts $B$ in (i)--(ii) can be chosen to be equal, and in (iii) we can choose an ELMD $Z$ for $\bbs$ with multiplicative decomposition $Z = D B^{-1}$ with this savings account $B$.
\end{theorem}

\begin{proof}
Let $B$ be an arbitrary savings account.

\noindent(i) $\Rightarrow$ (iii): By Proposition \ref{prop-NA-concepts-equiv} the set $\cali_0(\bbs B^{-1})$ also satisfies NA, and hence, by \cite[Thm. 1]{Kabanov-Stricker} the set $\cali_0(\bbs B^{-1}) - L_+^0$ is closed in $L^0$. Therefore, by \cite[Cor. 5.9]{Platen-Tappe-tvs} the set $\cali_0(\bbs B^{-1})$ also satisfies NFLVR. By \cite[Prop. 7.27]{Platen-Tappe-tvs} it follows that $\cali_1(\bbs B^{-1})$ satisfies NA$_1$, and hence NUPBR. Of course, the subset $\cali_1^+(\bbs B^{-1})$ also satisfies NUPBR. Therefore, by Proposition \ref{prop-NA-concepts-equiv} the set $\calp_{\sfi,1}^+(\bbs \cup \{ B \})$ satisfies NUPBR. Hence by Theorem \ref{thm-FTAP} there exists a local martingale $D > 0$ such that $Z = D B^{-1}$ is an ELMD for $\bbs$. By \cite[Thm. 1]{Jacod-Shiryaev-Finance} the process $D$ is a generalized martingale. Therefore, for each $t=0,\ldots,T$ we have $\bbp$-almost surely
\begin{align*}
\bbe[D_t] = \bbe[D_t | \calf_0] = D_0 < \infty,
\end{align*}
and hence $D_t \in \call^1$, proving that $D$ is a martingale. Analogously, we show that the processes $S^1 Z, \ldots, S^d Z$ are martingales, proving that $Z$ is an EMD for $\bbs$.

\noindent(iii) $\Rightarrow$ (ii): Note that $D$ is an EMD for $\bbs B^{-1}$. Let $\bbq \approx \bbp$ be the equivalent probability measure on $(\Omega,\calf_T)$ with Radon-Nikodym derivative $\frac{d \bbq}{d \bbp} = D_T / D_0$. Then $\bbq$ is an EMM for $\bbs B^{-1}$.

\noindent(ii) $\Rightarrow$ (i): Let $\xi \in \cali_0(\bbs B^{-1}) \cap L_+^0$ be arbitrary. Then there exists a strategy $\delta \in \Delta(\bbs B^{-1})$ such that $(\delta \bdot (S B^{-1}))_T = \xi$. Since $\bbq$ is an EMM for $\bbs B^{-1}$, the process $M := \delta \bdot (S B^{-1})$ is a $d$-martingale transform under $\bbq$. Therefore, by \cite[Thm. 1]{Jacod-Shiryaev-Finance} the process $M$ is a generalized $\bbq$-martingale. Hence, we have $\bbq$-almost surely
\begin{align*}
\bbe_{\bbq}[\xi] = \bbe_{\bbq}[M_T] = \bbe_{\bbq}[M_T | \calf_0] = M_0 = 0.
\end{align*}
Since $\xi \geq 0$ and $\bbq \approx \bbp$, we deduce that $\xi = 0$. Hence $\cali_0(\bbs B^{-1})$ satisfies NA, and by Proposition \ref{prop-NA-concepts-equiv} it follows that $\calp_{\sfi,0}(\bbs \cup \{ B \})$ satisfies NA.
\end{proof}

Important in the previous result is to note that NA is equivalent to the existence of an EMM $\bbq \approx \bbp$, which connects to the well-known no-arbitrage result in discrete time; see, e.g. \cite{FS}. This is due to the fact that, in the present discrete time setting, for every deflator as in Theorem \ref{thm-FTAP}, which is a multiplicative special semimartingale, the local martingale part is a true martingale, which gives rise to the aforementioned measure change. This finding indicates that without moving to continuous time modeling one would be unable to  exploit less expensive real world pricing that would work in practice when the existing market would have an ELMD that is a strict local martingale.

\begin{remark}
Concerning the existence of self-financing arbitrage portfolios, we can draw the following conclusions:
\begin{enumerate}
\item In the situation of Theorem \ref{thm-FTAP} a nonnegative self-financing arbitrage portfolio does not exist in the market $\bbs \cup \{ B \}$. Note that this does not exclude the existence of arbitrage portfolios which go negative in between, say admissible arbitrage portfolios.

\item In the situation of Theorem \ref{thm-FTAP-classical} an admissible self-financing arbitrage portfolio does not exist in the market $\bbs \cup \{ B \}$. Note that this does not exclude the existence of arbitrage portfolios. However, such an arbitrage portfolio cannot be admissible, and hence requires an unbounded credit line, i.e., it is not uniformly bounded from below.

\item In the situation of Theorem \ref{thm-diskret-2} a self-financing arbitrage portfolio does not exist in the market $\bbs \cup \{ B \}$.
\end{enumerate}
\end{remark}

\begin{appendix}

\section{Vector stochastic integration}\label{app-vector-integration}

In this appendix we provide the required results about vector stochastic integration. A general reference about this topic is \cite{Shiryaev-Cherny}, to which we also refer concerning upcoming notations.

Let $(\Omega,\calf,(\calf_t)_{t \in \bbr_+},\bbp)$ be a stochastic basis satisfying the usual conditions. For a multi-dimensional semimartingale $X \in \cals^d$ we denote by $L(X)$ the space of all $X$-integrable processes; see \cite{Shiryaev-Cherny}. The proof of the following result is straightforward and therefore omitted.

\begin{lemma}\label{lemma-int-same-comp}
Let $X \in \cals$ be a semimartingale, and let $H,K$ be two predictable $\bbr^d$-valued processes such that $H \cdot K \in L(X)$. Then we have $HK \in L(X \bbI_{\bbr^d})$ and the identity
\begin{align*}
(H \cdot K) \bdot X = (HK) \bdot (X \bbI_{\bbr^d}),
\end{align*}
where the $\bbr^d$-valued process $HK$ has the components $(HK)^i := H^i K^i$ for each $i=1,\ldots,d$, and where $X \bbI_{\bbr^d}$ denotes the $\bbr^d$-valued process $(X,\ldots,X)$.
\end{lemma}

\begin{lemma}\label{lemma-ass}
Let $X \in \cals^d$ and $H \in L(X)$ be arbitrary. Let $K$ be a $\bbr$-valued predictable, locally bounded process. Then we have
\begin{align*}
K \in L(H \bdot X), \quad KH \in L(X), \quad H \in L(K \bdot X)
\end{align*}
and the identities
\begin{align*}
K \bdot (H \bdot X) = (KH) \bdot X = H \bdot (K \bdot X),
\end{align*}
where $K \bdot X$ denotes the $\bbr^d$-valued process with components $(K \bdot X)^i := K \bdot X^i$ for each $i=1,\ldots,d$.
\end{lemma}

\begin{proof}
Since $K$ is predictable and locally bounded, we have $K \in L(H \bdot X)$, and by \cite[Thm. 4.6]{Shiryaev-Cherny} we obtain $KH \in L(X)$ and
\begin{align*}
K \bdot (H \bdot X) = (KH) \bdot X.
\end{align*}
Since $K$ is predictable and locally bounded, we also have $K \in L(X^i)$ for each $i=1,\ldots,d$. Since $KH \in L(X)$, by \cite[Thm. 4.7]{Shiryaev-Cherny} we obtain $H \in L(K \bdot X)$ and
\begin{align*}
H \bdot (K \bdot X) = (KH) \bdot X,
\end{align*}
completing the proof.
\end{proof}

\begin{lemma}\label{lemma-bracket}
Let $X \in \cals^d$, $Y \in \cals$ and $H \in L(X)$ be arbitrary. Then we have $H \in L_{\var}([X,Y])$ and the identity
\begin{align*}
[H \bdot X, Y] = H \bdot [X,Y],
\end{align*}
where $[X,Y] \in \calv^d$ denotes the $\bbr^d$-valued process with components $[X^i,Y]$ for each $i=1,\ldots,d$.
\end{lemma}

\begin{proof}
Using the notation from \cite[Thm. 4.19]{Shiryaev-Cherny}, we have $e=1$ and $K=1$. Let $F \in \calv^+$ and an optional $\bbr^d$-valued processes $\rho$ be such that
\begin{align*}
[X^i,Y] &= \rho^i \bdot F, \quad i=1,\ldots,d.
\end{align*}
By \cite[Thm. 4.19]{Shiryaev-Cherny} we obtain $H \cdot \rho \in L(F)$, which means that
\begin{align*}
| H \cdot \rho | \bdot F \in \calv^+,
\end{align*}
and hence $H \in L_{\var}([X,Y])$. Furthermore, by \cite[Thm. 4.19]{Shiryaev-Cherny} we have
\begin{align*}
[H \bdot X, Y] = ( H \cdot \rho ) \bdot F = H \bdot [X,Y],
\end{align*}
completing the proof.
\end{proof}

\section{Market transformations}\label{app-trans}

In this appendix we review a well-known transformation result for self-financing portfolios and draw some conclusions for no-arbitrage concepts. The mathematical framework is that of Section \ref{sec-notation}. We introduce the notation
\begin{align*}
\bar{\bbs} := \bbs \cup \{ 1 \}.
\end{align*}

\begin{lemma}\label{lemma-sf-bijection}
Suppose that $1 \notin \bbs$. Then there is a bijection between $\bbr \times \Delta(\bbs)$ and $\Delta_{\sfi}(\bar{\bbs})$. Furthermore, for all $(x,\delta) \in \bbr \times \Delta(\bbs)$ and the corresponding strategy $\bar{\delta} \in \Delta_{\sfi}(\bar{\bbs})$ we have
\begin{align*}
\bar{S}^{\bar{\delta}} = x + \delta \bdot S.
\end{align*}
\end{lemma}

\begin{proof}
This is a consequence of \cite[Lemma 5.1]{Takaoka-Schweizer}.
\end{proof}

For the next result, recall the notation (\ref{market-Y}).

\begin{lemma}\label{lemma-sf-num-tech}\cite[Prop. 5.2]{Takaoka-Schweizer}
Suppose that $1 \notin \bbs$. Let $\delta \in \Delta_{\sfi}(\bar{\bbs})$ be a self-financing strategy, and let $Y \geq$ be a nonnegative semimartingale. Then we also have $\delta \in \Delta_{\sfi}(\bar{\bbs} Y)$.
\end{lemma}

Recall that $\bbi_{\alpha}$ consists of all integral processes starting in $\alpha$.

\begin{lemma}\label{lemma-wealth-add-1}
Let $X \in \bbs$ be such that $X,X_- > 0$, and set $\bbs_0 := \bbs \setminus \{ X \}$. Then for each $\alpha \geq 0$ we have 
\begin{align*}
\bbi_{\alpha}(\bbs X^{-1}) = \bbi_{\alpha}(\bbs_0 X^{-1}).
\end{align*}
\end{lemma}

\begin{proof}
Noting that $\bbs X^{-1} = \bbs_0 X^{-1} \cup \{ 1 \}$, the result follows.
\end{proof}

\begin{lemma}\label{lemma-B-operation}
For each savings account $B$ we have $\overline{\bbs B^{-1}} B = \bbs \cup \{ B \}$.
\end{lemma}

\begin{proof}
We have
\begin{align*}
\overline{\bbs B^{-1}} B = ( \bbs B^{-1} \cup \{ 1 \} ) B = \bbs \cup \{ B \},
\end{align*}
completing the proof.
\end{proof}

\begin{proposition}\label{prop-NA-concepts-equiv}
Let $B$ be a savings account. Then the following statements are true:
\begin{enumerate}
\item $\calp_{\sfi,0}(\bbs \cup \{ B \})$ satisfies NA if and only if $\cali_0(\bbs B^{-1})$ satisfies NA.

\item $\calp_{\sfi,1}^+(\bbs \cup \{ B \})$ satisfies NUPBR if and only if $\cali_1^+(\bbs B^{-1})$ satisfies NUPBR.

\item Suppose that $B$ is bounded. If the set $\calp_{\sfi,0}^{\adm}(\bbs \cup \{ B \})$ satisfies NFLVR, then $\cali_0^{\adm}(\bbs B^{-1})$ satisfies NFLVR.

\item Suppose that $B^{-1}$ is bounded. If the set $\cali_0^{\adm}(\bbs B^{-1})$ satisfies NFL, then $\calp_{\sfi,0}^{\adm}(\bbs \cup \{ B \})$ satisfies NFL.
\end{enumerate}
\end{proposition}

\begin{proof}
Each of the six implications has a similar proof. Exemplarily, we shall prove the third statement. By Lemma \ref{lemma-wealth-add-1} we may assume that $B \notin \bbs$; otherwise we consider $\bbs_0 := \bbs \setminus \{ B \}$ rather than $\bbs$. Let $\xi \in \overline{\calc} \cap L_+^{\infty}$ be arbitrary, where
\begin{align*}
\calc := ( \cali_0^{\adm}(\bbs B^{-1}) - L_+^0 ) \cap L^{\infty}.
\end{align*}
Then there exists a sequence $(\xi^j)_{j \in \bbn} \subset \calc$ such that $\| \xi^j - \xi \|_{L^{\infty}} \to 0$. Let $j \in \bbn$ be arbitrary. Then there exist a strategy $\delta^j \in \Delta(\bbs B^{-1})$ and a constant $a^j \in \bbr_+$ such that
\begin{align*}
\delta^j \bdot ( S B^{-1} ) &\geq -a^j,
\\ \big( \delta^j \bdot ( S B^{-1} ) \big)_T &\geq \xi^j.
\end{align*}
By Lemma \ref{lemma-sf-bijection} there is a self-financing strategy $\bar{\delta}^j$ of the form
\begin{align*}
\bar{\delta}^j = (\delta^j,\eta^j) \in \Delta_{\sfi}(\overline{\bbs B^{-1}})
\end{align*}
for some predictable process $\eta^j$ such that
\begin{align*}
\delta^j \bdot ( S B^{-1} ) = (\delta^j,\eta^j) \cdot ( S B^{-1}, 1 ) = \delta^j \cdot (S B^{-1}) + \eta^j.
\end{align*}
Therefore, we have
\begin{align*}
\delta_0^j \cdot ( S_0 B_0^{-1} ) + \eta_0^j &= 0,
\\ \delta^j \cdot ( S B^{-1} ) + \eta^j &\geq -a^j,
\\ \delta_T^j \cdot ( S_T B_T^{-1} ) + \eta_T^j &\geq \xi^j.
\end{align*}
By Lemmas \ref{lemma-sf-num-tech} and \ref{lemma-B-operation} we have
\begin{align*}
\bar{\delta}^j \in \Delta_{\sfi}(\overline{\bbs B^{-1}} B) = \Delta_{\sfi}(\bbs \cup \{ B \}).
\end{align*}
Furthermore, we have
\begin{align*}
\delta_0^j \cdot S_0 + \eta_0^j \cdot B_0 &= 0,
\\ \delta^j \cdot S + \eta^j \cdot B &\geq -a^j B,
\\ \delta_T^j \cdot S_T + \eta_T^j \cdot B_T &\geq \xi^j B_T.
\end{align*}
In other words, we have
\begin{align}\label{B-bounded-1}
(S,B)_0^{\bar{\delta}^j} = 0, \quad (S,B)^{\bar{\delta}^j} \geq -a^j B \quad \text{and} \quad (S,B)_T^{\bar{\delta}^j} \geq \xi^j B_T.
\end{align}
Since $B$ is bounded, the portfolio $(S,B)^{\bar{\delta}^j}$ is admissible, and we have $\xi^j B_T \in L^{\infty}$. Therefore, we deduce that $\xi^j B_T \in \cale$, where
\begin{align}\label{B-bounded-2}
\cale := ( \calp_{\sfi,0}^{\adm}(\bbs \cup \{ B \}) - L_+^0 ) \cap L^{\infty}.
\end{align}
Since $\| \xi^j - \xi \|_{L^{\infty}} \to 0$ and $B_T \in L^{\infty}$, we also have $\| \xi^j B_T - \xi B_T \|_{L^{\infty}} \to 0$. Therefore, we have $\xi B_T \in \overline{\cale} \cap L_+^{\infty}$. Since $\calp_{\sfi,0}^{\adm}(\bbs \cup \{ B \})$ satisfies NFLVR, it follows that $\xi B_T = 0$, and hence $\xi = 0$. This proves that $\cali_0^{\adm}(\bbs B^{-1})$ satisfies NFLVR.
\end{proof}

\begin{remark}\label{remark-B-bounded}
In the proof of the third statement, the assumption that the savings account $B$ is bounded, is needed directly after (\ref{B-bounded-1}) in order to show that the portfolio $(S,B)^{\bar{\delta}^j}$ is admissible and that its terminal value belongs to the set $\cale$ introduced in (\ref{B-bounded-2}). For similar reasons, we require that $B^{-1}$ is bounded in the proof of the fourth statement.
\end{remark}

\section{Equivalent local martingale deflators and related concepts}\label{app-ELMD}

In this appendix we present the required results about local martingale deflators and related concepts. The mathematical framework is that of Section \ref{sec-notation}. In particular, recall the sets of potential security processes which we have introduced there. Now, we introduce the unions
\begin{align*}
\bbi(\bbs) := \bigcup_{\alpha \geq 0} \bbi_{\alpha}(\bbs), \quad \bbi^{\adm}(\bbs) := \bigcup_{\alpha \geq 0} \bbi_{\alpha}^{\adm}(\bbs) \quad \text{and} \quad \bbi^+(\bbs) := \bigcup_{\alpha \geq 0} \bbi_{\alpha}^+(\bbs).
\end{align*}

\begin{lemma}\cite[Cor. 3.5]{Ansel-Stricker}\label{lemma-M-sigma-loc}
For every admissible process $X \in \calm_{\sigma}$ we have $X \in \calm_{\loc}$.
\end{lemma}

\begin{proposition}\label{prop-Z-deflator}
For a semimartingale $Z$ with $Z,Z_- > 0$ the following statements are equivalent:
\begin{enumerate}
\item[(i)] $Z$ is an ELMD for $\bbs$, and we have $Z \in \calm_{\loc}$.

\item[(ii)] $Z$ is an ELMD for $\bbi^{\adm}(\bbs)$.

\item[(iii)] $Z$ is an E$\it{\Sigma}$MD for $\bbi(\bbs)$.
\end{enumerate}
\end{proposition}

\begin{proof}
(i) $\Rightarrow$ (iii): Let $\alpha \in \bbr_+$ and $\delta \in \Delta(\bbs)$ be arbitrary. By Proposition \ref{prop-deflator-prop-1} we have
\begin{align*}
I^{\alpha,\delta} Z = (\alpha + \delta \bdot S)Z = \alpha Z + (\delta \bdot S) Z \in \calm_{\sigma}.
\end{align*}

\noindent (iii) $\Rightarrow$ (ii): Noting that $\bbi^{\adm}(\bbs) \subset \bbi(\bbs)$, this implication follows from Lemma \ref{lemma-M-sigma-loc}.

\noindent (ii) $\Rightarrow$ (i): We have $\bbs \subset \bbi^{\adm}(\bbs)$. Therefore, the process $Z$ is an ELMD for $\bbs$. Furthermore, setting $\alpha := 1$ and $\delta := 0$ we obtain $I^{\alpha,\delta} = 1$, and hence
\begin{align*}
Z = I^{\alpha,\delta} Z \in \calm_{\loc},
\end{align*}
completing the proof.
\end{proof}

\begin{proposition}\label{prop-ELMM-W}
For an equivalent probability measure $\bbq \approx \bbp$ on $(\Omega,\calf_T)$ the following statements are equivalent:
\begin{enumerate}
\item[(i)] $\bbq$ is an ELMM for $\bbs$.

\item[(ii)] $\bbq$ is an ELMM for $\bbi^{\adm}(\bbs)$.

\item[(iii)] $\bbq$ is an E$\it{\Sigma}$MM for $\bbi(\bbs)$.
\end{enumerate}
\end{proposition}

\begin{proof}
(ii) $\Rightarrow$ (i): Since $\bbs \subset \bbi^{\adm}(\bbs)$, this implication is obvious.

\noindent(iii) $\Rightarrow$ (ii): Noting that $\bbi^{\adm}(\bbs) \subset \bbi(\bbs)$, this implication follows from Lemma \ref{lemma-M-sigma-loc}.

\noindent(i) $\Rightarrow$ (iii): Let $\alpha \in \bbr_+$ and $\delta \in \Delta(\bbs)$ be arbitrary. Then by \cite[Lemma 5.6]{Shiryaev-Cherny} the process
\begin{align*}
I^{\alpha,\delta} = \alpha + \delta \bdot S
\end{align*}
is a $\bbq$-$\sigma$-martingale.
\end{proof}

Now, we introduce the unions
\begin{align*}
\bbp_{\sfi}(\bbs) := \bigcup_{\alpha \geq 0} \bbp_{\sfi,\alpha}(\bbs), \quad \bbp_{\sfi}^{\adm}(\bbs) := \bigcup_{\alpha \geq 0} \bbp_{\sfi,\alpha}^{\adm}(\bbs) \quad \text{and} \quad \bbp_{\sfi}^+(\bbs) := \bigcup_{\alpha \geq 0} \bbp_{\sfi,\alpha}^+(\bbs).
\end{align*}

\begin{proposition}\label{prop-sf}
The following statements are equivalent:
\begin{enumerate}
\item[(i)] $Z$ is an ELMD for $\bbs$.

\item[(ii)] $Z$ is an ELMD for $\bbp_{\sfi}^{\adm}(\bbs)$.

\item[(iii)] $Z$ is an E$\it{\Sigma}$MD for $\bbp_{\sfi}(\bbs)$.
\end{enumerate}
\end{proposition}

\begin{proof}
(i) $\Rightarrow$ (iii): Let $\delta \in \Delta_{\sfi}(\bbs)$ be arbitrary. By Proposition \ref{prop-deflator-prop-2} we have
\begin{align*}
S^{\delta} Z = (\delta \cdot S) Z \in \calm_{\sigma}.
\end{align*}

\noindent (iii) $\Rightarrow$ (ii): Noting that $\bbp_{\sfi}^{\adm}(\bbs) \subset \bbp_{\sfi}(\bbs)$, this implication follows from Lemma \ref{lemma-M-sigma-loc}.

\noindent (ii) $\Rightarrow$ (i): Noting that $\bbs \subset \bbp_{\sfi}^{\adm}(\bbs)$, the process $Z$ is an ELMD for $\bbs$.
\end{proof}

\section{Sufficient conditions for the absence of arbitrage}\label{app-suff}

In this appendix we present a result containing sufficient conditions for the absence of arbitrage. The mathematical framework is that of Section \ref{sec-notation}.

\begin{proposition}\label{prop-ELMD-suff}
The following statements are true:
\begin{enumerate}
\item Suppose that an ELMM $\bbq \approx \bbp$ on $(\Omega,\calf_T)$ for $\bbs$ exists. Then $\cali_0^{\adm}(\bbs)$ satisfies NFL.

\item Suppose that an ELMD $Z$ for $\bbs$ with $Z \in \calm_{\loc}$ exists. Then $\cali_0^+(\bbs)$ satisfies NFL.

\item Suppose that an ELMD $Z$ for $\bbs$ exists. Then $\calp_{\sfi,1}^+(\bbs)$ satisfies NUPBR.
\end{enumerate}
\end{proposition}

\begin{proof}
Each of the three statements has a similar proof. Exemplarily, we shall prove the third statement. By Proposition \ref{prop-sf} the process $Z$ is also an ELMD for $\bbp_{\sfi}^+(\bbs)$. Let
\begin{align*}
\xi \in \bigcap_{\alpha > 0} \calb_{\alpha}
\end{align*}
be arbitrary, where
\begin{align*}
\calb_{\alpha} := ( \calp_{\sfi,\alpha}^+(\bbs) - L_+^0) \cap L_+^0 \quad \text{for each $\alpha > 0$.}
\end{align*}
Let $\alpha > 0$ be arbitrary. Then there exists a self-financing strategy $\delta^{\alpha} \in \Delta_{\sfi}(\bbs)$ such that
\begin{align*}
S_0^{\delta^{\alpha}} = \alpha, \quad S^{\delta^{\alpha}} \geq 0 \quad \text{and} \quad S_T^{\delta^{\alpha}} \geq \xi.
\end{align*}
Since $Z$ is an ELMD for $\bbp_{\sfi}^+(\bbs)$, the process $S^{\delta^{\alpha}} Z$ is a nonnegative local martingale, and hence a supermartingale. By Doob's optional stopping theorem for supermartingales we obtain
\begin{align*}
\bbe[\xi Z_T] \leq \bbe[S_T^{\delta^{\alpha}} Z_T] \leq \bbe[S_0^{\delta^{\alpha}} Z_0] = \alpha Z_0.
\end{align*}
Since $\alpha > 0$ was arbitrary, we deduce that $\bbe[\xi Z_T] = 0$. Since $\xi \geq 0$ and $\bbp(Z_T > 0)=1$, this shows $\xi = 0$. Therefore, by \cite[Thm. 7.25]{Platen-Tappe-tvs} the set $\calp_{\sfi,1}^+(\bbs)$ satisfies NUPBR.
\end{proof}

\end{appendix}

\bibliographystyle{plain}

\bibliography{Finance}

\end{document}